\newif\ifabstract
\abstracttrue
 \abstractfalse 
\newif\iffull
\ifabstract \fullfalse \else \fulltrue \fi

\documentclass[11pt]{article}

\advance \oddsidemargin by -0.8in
\newcommand{\myparskip}{3pt}
\parskip \myparskip

\usepackage{amsmath}
\usepackage{amsfonts}
\usepackage{amsthm}
\usepackage{algorithm}
\usepackage{algorithmic}
\usepackage{xcolor}
\usepackage{xspace}
\usepackage{subfigure}
\usepackage{hyperref}
\usepackage{caption}
\usepackage{graphicx} 
\usepackage[export]{adjustbox}
\usepackage{tabularray}

\newcommand{\caln}{\mathcal{N}}
\newcommand{\calt}{\mathcal{T}}

\newcommand{\cali}{\mathcal{I}}

\newcommand{\la}{{\lambda}}

\newcommand{\bx}{\boldsymbol{x}}
\newcommand{\by}{\boldsymbol{y}}

\newcommand{\bmu}{\boldsymbol{\mu}}
\newcommand{\bla}{\boldsymbol{\lambda}}

\newcommand{\bsigma}{\boldsymbol{\sigma}}
\newcommand{\bet}{\boldsymbol{\eta}}

\newcommand{\hy}{\hat{y}}
\newcommand{\hx}{\hat{x}}

\newcommand{\opt}{\texttt{OPT}\xspace}
\newcommand{\alg}{\texttt{ALG}\xspace}

\newcommand{\opa}{\texttt{OPA}\xspace}
\newcommand{\opd}{\texttt{OPD}\xspace}

\newcommand{\russell}[1]{  \ifthenelse{\boolean{showcomments}}
	{\textcolor{blue}{(Russell says:  #1)}}{}}

\newtheorem{theorem}{Theorem}
\newtheorem{lemma}[theorem]{Lemma}
\newtheorem{thm}{Theorem}
\newtheorem{assumption}{Assumption}

\newtheorem{pro}[thm]{Proposition}

\newcommand{\CR}{\texttt{CR}}

\graphicspath{{Images/}}

\begin{document}

\title{Near-optimal Online Algorithms for Joint Pricing and Scheduling in EV Charging Networks}

\def\thefootnote{$\#$}\footnotetext{These authors contribute equally to this work.}\def\thefootnote{\arabic{footnote}}

\author{Roozbeh Bostandoost$^\#$\thanks{University of Massachusetts Amherst. Email: {\tt  rbostandoost@umass.edu}.} \and
Bo~Sun$^\#$\thanks{The Chinese University of Hong Kong. Email: {\tt bsun@cse.cuhk.edu.hk}.}
\and 
\and
Carlee Joe-Wong\thanks{Carnegie Mellon University. Email: {\tt cjoewong@andrew.cmu.edu. } Partial support was provided by NSF CNS-2103024.}\and
Mohammad~Hajiesmaili\thanks{University of Massachusetts Amherst. Email: {\tt hajiesmaili@cs.umass.edu}. Partial support was provided by NSF CAREER-2045641, CNS-2102963, and CPS-2136199.} 
}

\begin{titlepage}
\maketitle

\thispagestyle{empty}

\begin{abstract}
With the rapid acceleration of transportation electrification, public charging stations are becoming vital infrastructure in a smart sustainable city to provide on-demand electric vehicle (EV) charging services. As more consumers seek to utilize public charging services, the pricing and scheduling of such services will become vital, complementary tools to mediate competition for charging resources. However, determining the right prices to charge is difficult due to the online nature of EV arrivals.
This paper studies a joint pricing and scheduling problem for the operator of EV charging networks with limited charging capacity and time-varying energy cost.
Upon receiving a charging request, the operator offers a price, and the EV decides whether to admit the offer based on its own value and the posted price. The operator then schedules the real-time charging process to satisfy the charging request if the EV admits the offer. 
We propose an online pricing algorithm that can determine the posted price and EV charging schedule to maximize social welfare, i.e., the total value of EVs minus the energy cost of charging stations. Theoretically, we prove the devised algorithm can achieve the order-optimal competitive ratio under the competitive analysis framework. Practically, we show the empirical performance of our algorithm outperforms other benchmark algorithms in experiments using real EV charging data. 
\end{abstract}
\end{titlepage}

\maketitle
\section{Introduction}
EVs have long promised to make transportation systems more efficient and sustainable~\cite{epa}. Bloomberg predicts that by 2040, more than half of the new car sales will be EVs, accounting for 5$\%$ of total electricity usage~\cite{bloomberg}.
According to the New York Times, California has already set a 2035 deadline for all new cars sold in California to be powered by electricity or hydrogen and be free of greenhouse gas emissions. The rule also requires that 35\% of new passenger vehicles sold by 2026 produce zero emissions increasing to 68\% of new vehicles.
This predicted EV popularity follows a similarly rapid growth in the number of public charging stations. For example, the city of Los Angeles has 4,689 public charging stations and has added 1185 new charging stations since mid June 2022. ChargePoint, the world's largest EV charging network, plans to build out 2.5 million public charging ports by 2025~\cite{electrek}. Despite the growing number of charging stations, however, each individual station has fairly limited charging capacity that may not be enough to serve the growing demand from EV vehicles. The scale of this growth then raises a fundamental challenge: \emph{how should EV charging networks manage the ever-increasing demands for their services?}

Answering this question naturally raises \emph{dynamic pricing} as the solution. To regulate their users' demands, charging platforms may wish to charge users more when demand is high relative to their limited capacity and vice versa when demand is low.
The advantages of dynamic pricing over static pricing (i.e., users are charged with the same price independent of user demand and station capacity) have been extensively discussed in theoretic studies~\cite{borenstein2002dynamic, al2013cloud, chakraborty2013dynamic}.
In practice, ChargePoint is already piloting such dynamic pricing~\cite{chargepoint}. 
Dynamic pricing allows these platforms to adjust their charging demands to match available resources,  
i.e., the transformers' capacity, while carefully scheduling which users receive resources at which time can further improve revenue and efficiency. Yet despite considerable prior work on both pricing and scheduling problems for adaptive EV charging networks~\cite{lin2021minimizing, zhao2015peak, 7914741}, there is still little rigorous theoretical understanding on how they can be jointly optimized to improve social welfare and ensure platform profitability.

Optimizing these prices for charging, however, requires resolving \emph{uncertainty challenges} of the future charging demand in the network. Moreover, the possibility of adaptive EV charging scheduling provides a new design space for further optimization. For example, the users may park their EVs at a charging platform while at work, and the vehicle can be charged at any time before the user leaves work. Thus, the EV's impact on the platform, and the price it should be charged, will depend on how congested the platform is at any time before it leaves. However, users would generally wish to know the prices and their guaranteed charging amount upon their arrivals in order to eliminate payment uncertainty (e.g., as offered by ride-sharing services~\cite{slavulj2016evolution}), at which time future information on platform congestion is unknown.

In this paper, we tackle the problem of joint pricing and scheduling for EV charging from the perspective of a charging operator, who manages a charging station. Our goal is to develop pricing algorithms for admitting EV users to the charging network and scheduling their charging rate such that their energy demand is fully satisfied during their window of availability, and this windows means the time interval that the EV is present at the charging station.
A major challenge in the above pricing and scheduling problem is the uncertainty of the environment in terms of future EV charging demand. EVs arrive at the charging station in an online manner and their availability windows might be different.

The rest of the paper is organized as follows. Section~\ref{sec:prob_statement} introduces the offline and online formulation of the joint pricing and scheduling problem. The proposed online algorithm and our main theoretical result are presented in Section~\ref{sec:alg&res}.
Then Section~\ref{sec:comp_analysis} provides proofs of the key theorems in Section~\ref{sec:alg&res}. Section~\ref{sec:exp_res} validates numerical performance of our proposed algorithm using real data.
Finally, we discuss the related work Section~\ref{sec:related_works} and conclude the paper in Section~\ref{sec:concludsion}.


\subsection{Contributions}
In this paper, we formulate the joint problem of dynamic pricing and scheduling for EV charging in a station with multiple charging ports and limited capacity. In this scenario, some EVs with different private values for getting charge, energy demand, and availability window arrive in \textit{an online manner}, and the objective is to maximize the social welfare by offering a charging price with proper scheduling, given the capacity limit of the station. 

We develop \textit{online algorithms} for the above joint pricing and scheduling problem. Upon arrival of a charging request from an EV, the algorithm calculates a charging scheduling for the EV, and posts a corresponding total charging price for the EV. 
The posted price is determined as a function of the real-time utilizations of the charging station, where the pricing function is designed by solving a differential equation induced by an online primal-dual approach.
The EV compares the posted price from the station and if the price is less than its private value, the value that the EV is willing to pay, it accepts the offer and joins the station. The calculation of the price is based on the charging request of the EV, e.g., energy demand, the availability windows; the utilization of the station, and the time-vary energy cost of the station. 

We then analyze the robustness of the proposed algorithms using competitive ratio (CR) as a well-established performance metric for online algorithms. Our analysis in Theorem~\ref{theo:thresh_func} shows that our algorithm is $O(\ln\theta)$-competitive, where  ${\theta = U D^{\max}/(D^{\min}(L - p^{\max}))}$, $U$ and $L$ are the maximum and minimum private value of per-unit energy requested by the EVs, $D^{\max}$ and $D^{\min}$ are the maximum and minimum availability window of EVs, and $p^{\max}$ is the maximum energy cost.
By establishing a connection between the problem of interest in this paper and two classic online problems in the prior literature, we then derive a lower bound of $\Omega(\ln (UD^{\max}/LD^{\min}))$ for the CR of any algorithm for the joint pricing and scheduling problem. This lower bound result shows that the CR of the proposed algorithm is order-optimal.

Last, we use the Caltech EV ACN-Data dataset~\cite{lee2019acn} on over 700 EV charging sessions over 90 days to evaluate the performance of the proposed algorithm and compare it with the prior algorithms and baselines. The experimental results show that the empirical CR of the proposed algorithm is less than 2 in 80\% of the experiments on average, which is much better than our theoretical bounds. Also, this empirical result shows that our algorithm outperforms other comparison algorithms substantially; the empirical CRs of others is less than 2 in only 65\%, 25\%, and 20\% of the experiments on average.

\begin{table*}
\captionof{table}{Comparison of previous works and ours in terms of multiple criteria}
\vspace{-3mm}
\label{tab:comp_intro}
\footnotesize
\begin{tabular} {|c | c c c c | c|}
 \hline
  & \begin{tabular}{@{}c@{}}\textbf{Integral/} \\ \textbf{Fractional}\end{tabular} & \textbf{Scheduling} &  \textbf{Reusable Resource} & \textbf{Resource Cost} & \textbf{Optimality of CRs}\\ [0.5ex] 
 \hline \hline
 \cite{zhou2008budget}
 & Integral & No & No & No & Yes\\
 \hline
 \cite{zheng2014online}
 & Fractional & Yes & Yes & Yes & N/A\\
 \hline
 \cite{sun2020competitive} & Fractional & Yes & Yes & No & Yes\\
  \hline
 \cite{zhang2017optimal} & Integral & Yes & Yes & No & No\\
 \hline
 \cite{tan2020mechanism} & Integral & No & Yes & Yes & No\\
  \hline
 \cite{sun2022online} & Integral & No & Yes & No & Yes\\ 
 \hline
 \textbf{This Work} & Integral & Yes & Yes & Yes & Yes\\
 \hline
\end{tabular}
\end{table*}

\subsection{The Significance of the Theoretical Results}
We note that although we present our joint pricing and scheduling algorithms in the context of EV charging, the underlying optimization problem is of independent interest and could be applicable to scenarios beyond EV charging. More specifically, the problem of interest in this paper is as an extension of the basic online knapsack problem and captures reusable resources with time-varying resource cost. Consider an online scenario where some items with different values and demand arrive one-by-one over time and stay in the system for a limited duration. Upon arrival of a new item, an online decision maker should decide whether or not to admit the item given its value and demand, the limited capacity of the knapsack, and the time-varying cost of using the knapsack. Further, if the decision is to admit the item, what is the scheduling policy to fully satisfy the demand of the admitted item. 

The online knapsack problem and its variants have been extensively studied in the literature~\cite{sun2020competitive,zhang2017optimal,tan2020mechanism,sun2022online,zheng2014online,zhou2008budget}.
Nearly all online algorithms in this stream of works are based on a similar idea that estimates the price of admitting one item using a function of the knapsack utilization and admits the item if item value is larger than the estimated price. The key of these algorithms lies in how to design the pricing function to guarantee the best possible CR.
We review the literature in Section~\ref{sec:related_works} extensively and in what follows we highlight the differences and significance of our results with the most relevant works.

The basic version of the online knapsack problem without reusable resource and resource cost was first studied in~\cite{zhou2008budget} and an online algorithm was designed to achieve the optimal CR. 
In the general setting, this stream of works can be distinguished based on three criteria: (i) fractional or integral admission decision, i.e., whether each item can be fractionally admitted or not; (ii) scheduling, i.e., whether the demand of the item is fixed or can be flexibly determined over time; and (iii) resource cost, i.e., whether using the knapsack incurs a cost.  In terms of the significance of the result, the ultimate goal is to achieve online algorithms that are \textit{(order)-optimal}, where optimality refers to an online algorithm with the best possible CR either exact or order wise. 
In Table~\ref{tab:comp_intro}, we compare the most related works based on these criteria. 

This paper designs an order-optimal online algorithm for the most general integral setting. 
Compared to the other two works with optimal CRs (i.e.,~\cite{sun2020competitive} and~\cite{sun2022online}), this work makes additional technical contributions in design and analysis of the algorithm in the general setting. 
In particular, we extend the analysis of the fractional admission in~\cite{sun2020competitive} to integral admission by considering two classes of instances and analyzing their corresponding worst-case ratios using different approaches (see Section~\ref{proof:cap_free_comp_ratio} and~\ref{proof:cap_lim_comp_ratio}). 
We extend the fixed demand setting in~\cite{sun2022online} to scheduling setting by formulating an auxiliary cost minimization (see problem~\eqref{p:cost-minimization}) for scheduling decisions and analyzing the CR by an online primal-dual approach. Further, our primal-dual analysis results in a natural design of the pricing function, which is in contrast to the ad-hoc design in~\cite{sun2022online}. In addition, this paper takes into account the resource cost, which adds an extra dimension in the algorithm analysis.

\begin{table}[t]
\centering
\captionof{table}{Notations}
\label{tb:sum_not}
\begin{tabular}{|l | l|}
 \hline
 \textbf{Name} & \textbf{Description} \\ [0.5ex] 
 \hline\hline
 $a_n$ & Arrival time of EV $n$  \\ 
 \hline
 $d_n$ & Departure time of EV $n$  \\ 
 \hline
 $E_n$ & Energy demand of EV $n$  \\ 
 \hline
 $R_n$ & Rate limit of charging power in one time slot for EV $n$  \\ 
 \hline
 $v_n$ & Valuation of EV $n$ for receiving its demand  \\ 
 \hline\hline
 $\caln$ & Set of all EVs  \\ 
 \hline
 T & Number of time slots  \\ 
 \hline
 $\calt$ & \{1, 2,.., T\}  \\ 
 \hline
 $\calt_n$ & Set of time slots that EV $n$ is available at the station  \\ 
 \hline\hline
 C & Charging station capacity \\ 
 \hline
 $p_t$ & Electricity price at time $t$ \\
 \hline\hline
 $y_{nt}$ & The amount of energy that is being delivered to EV $n$ at time $t$ \\
 \hline
 $y_n$ & $\{y_{nt}\}_{t \in \calt_n}$ \\
 \hline
 $x_n$ & Indicator variable which indicates whether EV $n$ has decided\\
  & to be charged at the station or not \\
  \hline
\end{tabular}
\end{table}

\section{Problem Statement} \label{sec:prob_statement}

{
We consider an online pricing and scheduling problem for EV charging management.
A set $\caln:=\{1,\dots,N\}$ of EVs sequentially arrive over a time horizon $\calt:= \{1,\dots,T\}$.
A system operator manages a charging station with multiple chargers and a fixed overall charging rate capacity $C$, and faces a time-varying electricity price $\{p_{t}\}_{t\in\calt}$. Note that we assume that there are sufficient chargers in the station to admit all EVs; however, the per-slot aggregate charging rate is limited to capacity $C$.
Upon arrival, each EV $n$ proposes a request $\gamma_{n} := \{a_{n}, d_{n}, E_{n}, R_{n}\}$, where $a_{n}$ and $d_{n}$ denote the arrival and departure time of EV $n$, $E_n$ is the energy demand, and $R_{n}$ is the rate limit of charging power in one time slot. For the charging request $\gamma_n$, EV $n$ has a private value $v_{n}$ (i.e., willingness-to-pay), which is unknown to the operator. 

After receiving the request $\gamma_n$, the operator posts a price $\xi_{n}$ for serving EV $n$.
Then EV $n$ itself decides to pay $\xi_{n}$ for charging if and only if EV $n$ has a positive utility surplus, i.e.,
    $v_{n} - \xi_{n} \ge 0$;
otherwise, EV $n$ will leave the system without charging and seek for alternative refueling opportunities.

If EV $n$ decides to join the station, the operator collects payment $\xi_{n}$ and schedules its charging $\by_{n}:= \{y_{nt}\}_{t\in\calt_{n}}$, where $\calt_{n}:=\{a_{n},a_{n}+1,\dots,d_{n} - 1\}$ is the set of time slots when EV $n$ is available and $y_{nt}$ is the delivered energy in slot $t$. The charging schedule incurs energy cost $\sum_{t\in\calt_{n}}p_{t}y_{nt}$.
For a quick reference, notations are summarized in Table~\ref{tb:sum_not}.

\noindent\textbf{Offline social welfare maximization problem.} Let $x_{n} \in \{0,1\}$ indicate whether EV $n$ decides to charge at the station. The utility surplus of all EVs is $\sum_{n\in\caln}(v_{n} - \xi_{n}) x_{n}$ and the surplus of the charging station operator is $\sum_{n\in\caln} (\xi_{n} - \sum_{t\in\calt_{n}}p_{t}y_{nt})x_n$. 
Our objective is to maximize the social welfare of all EVs and the operator~\cite{zheng2014online,sun2018eliciting}. 
Define an instance $\cali:= \{\{\gamma_{n}; v_{n}\}_{n\in\caln}\}$ as a sequence of the EV charging requests and the corresponding values.
Given $\cali$ from the start, the offline social welfare maximization problem can be formulated as follows.
\begin{subequations}\label{p:social-welfare}
    \begin{align}
    \max_{x_n, y_{nt}} \quad& \sum_{n\in\caln}v_nx_n - \sum_{t\in\calt} p_t \sum_{n\in\caln_t}y_{nt} \nonumber\\
    {\rm s.t.}\quad& \sum_{t\in\calt_n} y_{nt} \ge E_n x_n,\quad \forall n\in\caln, \label{cnst:full_demand} \\
    & \sum_{n\in\caln_t}y_{nt} \le C, \quad \forall t\in\calt, \label{cnst:capacity}\\\
    & 0\le y_{nt} \le R_n x_n, \quad \forall n\in\caln,t\in\calt_n,\label{cnst:rate}\\
    & x_n \in \{0,1\}, \quad \forall n\in\caln \label{cnst:admission},
    \end{align}
\end{subequations}
where $\caln_t$ denotes the set of EVs that are available at time $t$. Problem~\eqref{p:social-welfare} captures multiple constraints. Constraint~\eqref{cnst:full_demand} requires the charging station to \textit{fully} charge the EV $n$ up to its demand if it accepts the charging offer ($x_n=1$). Constraint ~\eqref{cnst:capacity} says the aggregate of scheduled energy for EVs that are available at each slot $t$ must be smaller than the station capacity. Constraint~\eqref{cnst:rate} requires the scheduled energy in each slot $t$ is less than the EV's rate limit.
}

\noindent\textbf{Online problem.}
We consider an online version of the social welfare maximization problem~\eqref{p:social-welfare}. 
Upon the arrival of each EV $n$, the operator must immediately decide the posted price $\xi_{n}$ only based on the past and current EV requests, and without knowing the future EV requests.  In this paper, we consider the electricity price $\left\{p_t\right\}_{t\in\calt}$ is known to the charging station operator from the start, e.g., predetermined time-of-use prices or estimated electricity price based on historical data. 
The electricity price and charging station capacity $C$ are defined as fixed setup information of the system.
Let $\opt(\cali)$ and $\alg(\cali)$ denote the social welfare achieved by the optimal offline problem and an online algorithm under instance $\cali$, respectively. An online algorithm is called $\pi$-competitive if $\opt(\cali)/\alg(\cali) \le \pi$ holds for all instances, where $\pi\ge1$ is called the competitive ratio (CR). 
We aim to design an online algorithm with a CR that is as small as possible.

\noindent\textbf{Assumptions.}
We make following assumptions through the paper.
\begin{assumption}[Infinitesimal schedule]\label{assumption1}
The rate limit of individual EV charging rate is much smaller than the charging capacity of the station, i.e., $R_{n} \ll C_, \forall n \in \caln$.
\end{assumption}
Practically, a charging station can charge tens of EVs simultaneously; therefore, its capacity is much larger than the charging rate of each EV. This assumption is also common in the literature with competitive analysis, e.g.,~\cite{sun2020competitive, huang2019welfare}. Additionally, assumption~\ref{assumption1} allows us to focus on the nature of our problem with mathematical convenience. 

Next, we define the value density of EV $n$ as its value on per-unit energy, i.e., $v_n/E_n$, which is the payment that each EV is willing to pay for per unit of energy.

\begin{assumption}[Bounded value density]\label{assumption2}
The value densities of EVs are bounded, i.e., $\frac{v_{n}}{E_n} \in [L, U], \forall n\in\caln$, where $L$ and $U$ are the lower and upper bounds for the value density.
\end{assumption}
In our algorithm design and analysis, we assume that $L$ and $U$ are known in advance, e.g., from user surveys or historical behavior. 

\begin{assumption}[Bounded availability]\label{assumption3}
The available window of EV charging is bounded, i.e., $d_{n} - a_{n} \in [D^{\min},  D^{\max}], \forall n \in \caln$.
\end{assumption}
In practice, it is reasonable to have a lower bound for the EV available window since the EV cannot stay less than a minimum duration needed to get charged up to their demand given the charging rate limit. Also, each EV's availability time typically has an upper bound that is imposed by the station. 

\begin{assumption}[Bounded electricity Price]\label{assumption4}
The electricity price $\{p_t\}_{t\in\calt}$ is bounded, i.e., $0 \le p_{t} \le p^{\max} \le L, \forall t \in \calt$.
\end{assumption}

We assume the value density is greater than the maximum electricity price, i.e., $v_n/E_n > p^{\max}, \forall n \in \mathcal{N}$. 
Thus, EVs can be charged at the stations with non-negative utility surplus.

\section{Algorithm and Results} \label{sec:alg&res}
In this section, we propose an online posted-pricing algorithm (\opa) that can jointly determine the posted price and the corresponding charging schedule in EV charging networks.
\subsection{An Online Posted-Pricing Algorithm}
\begin{algorithm}[!t]
	\caption{Online Posted-Pricing Algorithm ($\opa(\phi)$)}
	\label{alg:ppa}
	\begin{algorithmic}[1]
		\STATE \textbf{Inputs:} setup information $\{C, \{p_t\}_{t\in\calt}\}$, pricing function $\phi = \{\phi_t(\cdot)\}_{t\in\calt}$; 
		\STATE \textbf{Initialization:} capacity utilization $w_{t}^{(0)} = 0,\forall t\in\calt$; 
		\WHILE{a new EV $n$ arrives with $ \{a_n,d_n,E_n,R_{n}\}$}
		\STATE solve a candidate charging schedule $\{\hat{y}_{nt}\}_{t\in\calt_{n}}$ from
		\begin{subequations}\label{p:cost-minimization}
		\begin{align}
		\min_{y_{nt}\ge0} \quad & \sum\nolimits_{t \in \calt_{n}} \int_{w_{t}^{(n-1)}}^{w_{t}^{(n-1)}+y_{nt}}\phi_{t}(u)du\\
		\label{p:cost-minimization-marginal-cost}
		{\rm s.t.} \quad & \sum\nolimits_{t \in \calt_{n}} y_{nt} \ge E_n, \quad (\hat{\mu}_n)\\
		&y_{nt} \le R_{n}, \forall t\in\calt_{n}. \quad (\hat{\sigma}_{nt})
		\end{align}
		\end{subequations}
        \STATE set the posted price $\xi_{n} = \sum_{t\in\calt_{n}}\phi_{t}\left(w_{t}^{(n-1)}+\hat{y}_{nt}\right)\hat{y}_{nt}$;
        \IF{EV $n$ accepts the charging}
		\STATE 
		set $\hat{x}_{n} = 1$ and charge EV $n$ by $\{\hat{y}_{nt}\}_{t\in\calt_{n}}$;
		\ELSE
		\STATE set $\hat{x}_{n} = 0$, $\hat{y}_{nt} = 0, \forall t\in\calt_{n}$;
		\ENDIF
		\STATE update utilization $w_{t}^{(n)}=\begin{cases}
		w_{t}^{(n-1)} + \hat{y}_{nt} & t\in\calt_{n}\\
		w_{t}^{(n-1)} & t\in \calt\setminus \calt_{n}
		\end{cases}$.
		\ENDWHILE
	\end{algorithmic}
\end{algorithm}

When a new EV $n$ arrives, $\opa(\phi)$ uses a predetermined pricing function $\phi$ to estimate the charging cost of the EV and determines a candidate charging schedule $\{\hat{y}_{nt} \}_{t\in\calt_{n}}$ by solving a (pseudo) cost-minimization problem~\eqref{p:cost-minimization}. In this problem, $\phi_t(u)du$ estimates the (pseudo) cost of charging $du$ unit of energy when the utilization of the station is $u$ in time slot $t$.
Then the objective of the problem~\eqref{p:cost-minimization} is to minimize the total (pseudo) cost of satisfying the EV's request.
Then the algorithm sets the posted price as the total estimated cost of charging EV $n$, i.e., $\xi_{n} = \sum_{t\in\calt_{n}}\phi_{t}(w_{t}^{(n-1)}+\hat{y}_{nt})\hat{y}_{nt}$, where $w_{t}^{(n-1)}$ is the charging station's utilization in slot $t$ when a new EV $n$ arrives.
If EV $n$'s value is no smaller than the posted price, the offer will be accepted, and the station charges the EV based on the candidate schedule. 
Finally, the charging station's utilization is updated and used for estimating the charging cost for next EV.

The performance of $\opa(\phi)$ is determined by the pricing function. In this paper, we consider the pricing function in the following form: 
\begin{align} \label{eq:pricing-function}
    \phi_{t}(w) =
    \begin{cases}
    L, & w \in [0, \beta)\\
    \varphi(w), & w \in [\beta, C]\\
    +\infty, & w\in (C,+\infty)
    \end{cases}, \forall t\in\calt.
\end{align}
This function consists of three segments. The first flat segment ($w\in[0,\beta)$) is set to the lower bound of EVs' value density. Then if EV $n$ is among the first few EVs that come to the station, it will receive a posted price $E_n L$, which the lowest possible price, and thus will accept the offer. This design ensures that the station can at least admit some EVs to secure certain profits.
The second segment $\varphi(w)$ is a non-decreasing function to realize an intuitive idea that the station becomes more selective to admit EVs as its utilization increases.
The last segment sets the price to $\infty$ when the utilization exceeds the capacity. In this way, any charging schedule that exceeds the capacity constraint leads to an infinitely large posted price (which can be set to a cap price in practice), preventing capacity violations. 

\subsection{Main Results}
Let $\CR(\phi)$ denote the CR of \opa with pricing function $\phi$.
In this section, we present the main theoretical results of this paper, which show how to design the pricing function $\phi$ such that \opa can achieve the order-optimal CR for the joint pricing and scheduling problem.
\begin{theorem} \label{theo:thresh_func}
Under Assumptions~\ref{assumption1}-\ref{assumption4}, when the parameter satisfies ${(U/L)(D^{\max}/D^{\min}) \ge 2}$ and the pricing function is given by $\phi^*=\{\phi_t^*(\cdot)\}_{t\in\calt}$, where for all $t\in\calt$,
\begin{equation} \label{eq:thresh_func}
    \phi_{t}(w) =
    \begin{cases}
    L, & w \in [0, C/\alpha)\\
    \frac{L - p_{t}}{e} e^{\frac{\alpha}{C} w} + p_{t}, & w \in [C/\alpha, C]\\
    +\infty, & w\in (C,+\infty)
    \end{cases}, 
\end{equation}
with $\alpha = 1 + 2\ln(\theta)$ and $\theta = \frac{U D^{\max}/D^{\min}}{L - p^{\max}}$, 
then the competitive ratio of \opa is $\CR(\phi^*) = O(\ln(\theta))$.
\end{theorem}
This theorem provides the design of the pricing function for \opa and the corresponding CR. Note that the technical assumption ${(U/L)(D^{\max}/D^{\min}) \ge 2}$ states that multiplication of the value and duration fluctuation ratios should be larger than 2.
We postpone the proof of Theorem~\ref{theo:thresh_func} to Section~\ref{sec:comp_analysis}.

\begin{theorem}\label{theo:thresh_lower_bound}
There is no online algorithm that can achieve a competitive ratio smaller than $\Omega(\ln(\theta))$ for our problem.
\end{theorem}

\begin{proof} \label{app:lower-bound}
To prove we consider two special cases of our problem. First, the basic online knapsack problem~\cite{zhou2008budget} is a special case of our problem with no energy price, setting the departure time of each EV $n$ to $T$, the demand of $(T - a_n)R_n$, which correspond to $R_n$ charging demand in each slot.  Define $\rho = \frac{U}{L}$ as the fluctuation ratio, and the lower bound of online knapsack is $\Omega(\ln(\rho))$~\cite{zhou2008budget}.

Secondly, with the special case of homogeneous values, i.e., $v_n =v,\ \forall n$, no energy cost, i.e., $p_t =0, \ \forall t$, and different duration of EVs, our problem degenerates to the online interval scheduling problem~\cite{goyal2020online}. Let $\delta = D^{\max}/D^{\min}$ be the duration ratio. It is known that the lower bound on the competitive ratio for any online algorithm for the online interval scheduling is $\Omega(\ln(\delta))$~\cite{goyal2020online}. Now, with two special cases of the online knapsack problem and the online interval scheduling, the lower bound of any online algorithms for problem~\eqref{p:social-welfare} is
    \begin{align}
        \texttt{CR(OPA)} &\ge  \max \{\Omega(\ln(\delta)), \Omega(\ln(\rho))\} \ge \frac{1}{2}\Omega(\ln(\delta)) + \frac{1}{2}\Omega(\ln(\rho)) \notag \\
        &\ge \frac{1}{2} \Omega(\ln(\delta \rho)) = \Omega(\ln(\theta)), \notag
    \end{align}
where $\theta = \frac{U D^{\max}/D^{\min}}{L - p^{\max}}$ and with $p^{\max} = 0$, we have $\theta = \delta\rho$.
\end{proof}

Theorem~\ref{theo:thresh_lower_bound} gives a lower bound of the CR for the joint pricing and scheduling problem. Combining Theorem~\ref{theo:thresh_func} and Theorem~\ref{theo:thresh_lower_bound} concludes that our proposed \opa can achieve an order-optimal CR.

\section{Competitive Analysis} \label{sec:comp_analysis}
This section presents the proof of Theorem~\ref{theo:thresh_func}, and simultaneously shows how to design the threshold function $\phi^*$ for \opa to achieve the order-optimal CR.

\subsection{Roadmap for the Proof}
\label{sec:proof-roadmap}

Recall that in \opa, the pricing function becomes infinitely large once the utilization exceeds the capacity of the charging station, i.e., $\phi_t(w_t^{(n-1)}+\hy_{nt}) = +\infty, \forall t\in\calt, n\in\caln$ if $w_t^{(n-1)}+\hy_{nt} > C$.
Thus, if any EV's charging candidate schedule exceeds the station capacity, \opa sets the posted price to $+\infty$ 
to reject this request.
However, the discontinuity of the pricing function, $\phi_t(w): \mathbb{R}^+\to\mathbb{R}^+$ at $w = C$, results in difficulties in the competitive analysis since the cost minimization problem~\eqref{p:cost-minimization} may not be a convex problem, prohibiting from directly applying existing approaches (e.g., primal-dual based analysis).
Thus, we divide the analysis of \opa into to two cases.
Let $\Psi$ denote the set of instances that satisfy Assumptions~\ref{assumption1}-\ref{assumption4} and divide $\Psi$ into two disjoint subsets $\Psi^1$ and $\Psi^2$. $\Psi^1$ includes the instances, under which \opa does not output a candidate schedule that violates the capacity. Then the utilization of the charging station under $\Psi^1$ is away from the capacity $C$ and hence we call $\Psi^1$ \textit{capacity-free instances}.
$\Psi^2 = \Psi \setminus \Psi^1$ is then called \textit{capacity-limited instances} because there exists at least one time slot when the capacity can be violated by one candidate schedule. 

In the proof of Theorem~\ref{theo:thresh_func}, we first analyze the CR of $\opa(\phi^*)$ under capacity-free instances $\Psi^1$ in Section~\ref{proof:cap_free_comp_ratio}. In this case, the cost minimization problem~\eqref{p:cost-minimization} is a convex problem and we analyze the CR of \opa based on online primal-dual (\opd) framework~\cite{buchbinder2009design,tan2020mechanism}. In fact, we can provide a constructive proof, which not only proves CR, but also shows how to design the pricing function $\phi^*$.  
\begin{lemma}\label{lem:cap_free_comp_ratio}
Given pricing function $\phi^*$ in Equation~\eqref{eq:thresh_func}, the CR of $\opa(\phi^*)$ under capacity-free instances $\Psi^1$ is $\frac{\opt(\cali)}{\alg(\cali)} \le \alpha = 1 + 2\ln(\theta), \forall \cali\in\Psi^1$.
\end{lemma}

Given the pricing function $\phi^*$ obtained in the capacity-free case, we then analyze the CR of $\opa(\phi^*)$ under the capacity-limited instances $\Psi^2$. Since we cannot rely on the convexity of the problem~\eqref{p:cost-minimization} in this case, we derive the CR by figuring out the worst-case instances in different scenarios and finally obtain the following upper bound for the CR.

\begin{lemma}\label{lem:cap_lim_comp_ratio}
Given pricing function $\phi^*$ in Equation~\eqref{eq:thresh_func}, the competitive ratio of $\opa(\phi^*)$ under capacity-limited instances $\Psi^2$ is $\frac{\opt(\cali)}{\alg(\cali)} \le 3\max\left\{2\sqrt{e}, \frac{\theta\alpha}{\exp(\alpha/2 - 1)}\right\} = 3\sqrt{e}\max\left\{2, 1+ 2\ln\theta \right\}, \forall \cali\in\Psi^2$.
\end{lemma}

Combining Lemma~\ref{lem:cap_free_comp_ratio} and Lemma ~\ref{lem:cap_lim_comp_ratio} gives
\begin{align*}
    \CR(\alpha^*) &\le \max_{\cali\in\Psi^1\cup\Psi^2}\frac{\opt(\cali)}{\alg(\cali)}\\
    &= \max\left\{\alpha, 6\sqrt{e}, \frac{3\theta\alpha}{\exp(\alpha/2 - 1)}\right\} = O(\ln\theta),
\end{align*}
which completes the proof of Theorem~\ref{theo:thresh_func}.

\subsection{Proof of Lemma~\ref{lem:cap_free_comp_ratio}} \label{proof:cap_free_comp_ratio}

Based on \opd, we can derive the following technical lemma that provides a sufficient condition on the pricing function to ensure the competitiveness of \opa.    
\begin{lemma}
\label{eq:threshfunc1}
Under capacity-free instances, $\texttt{OPA}(\phi)$ is $\alpha$-competitive if the pricing function $\phi = \{\phi_t\}_{t\in\calt}$ is given by, for all $t\in\calt$, 
\begin{align}
    \phi_t(w) =
    \begin{cases}
    L, & w \in [0, \beta),\\
    \varphi_t(w), & w \in [\beta, C],
    \end{cases} 
\end{align}
where $\beta \ge C/\alpha$ is a utilization threshold, and $\varphi_t$ is a non-decreasing function that satisfies:
\begin{align}
    \begin{cases}\label{eq:thresh_ineq1}
    \alpha \varphi_t(w) - C \varphi_t'(w) \geq \alpha p_t, w \in [\beta, C] \\
    \varphi_t(\beta) = L, \varphi_t(C) \ge \frac{U D^{\max}}{D^{\min}}
    \end{cases}.
\end{align}
\end{lemma}

Based on Lemma~\ref{eq:threshfunc1}, we can design the pricing function $\phi$ that satisfies the differential equation~\eqref{eq:thresh_ineq1} and minimize the competitive ratio $\alpha$. 
To do so, we bind all inequalities in~\eqref{eq:thresh_ineq1} and solve this boundary value problem. Then we can obtain that $\beta = C/\alpha$ and the pricing function $\phi^*$ is given by equation~\eqref{eq:thresh_func}, which proves Lemma~\ref{lem:cap_free_comp_ratio}.
In what follows, we prove the technical Lemma~\ref{eq:threshfunc1} based on \opd.

Consider the relaxed primal problem~\eqref{p:social-welfare} as follows
\begin{subequations}\label{prelim_relaxed_prim}
    \begin{align}
    \max_{x_n\ge0, y_{nt}\ge 0} \quad& \sum_{n\in\caln}v_nx_n - \sum_{t\in\calt} p_t \sum_{n\in\caln_t}y_{nt}\\
    {\rm s.t.}\quad& \sum_{t\in\calt_n} y_{nt} \ge E_n x_n,\quad \forall n\in\caln, \quad (\mu_n)\\
    & \sum_{n\in\caln_t}y_{nt} \le C, \quad \forall t\in\calt, \quad (\lambda_t)\\\
    & y_{nt} \le R_n x_n, \quad \forall n\in\caln,t\in\calt_n, \quad (\sigma_{nt})\\
    & x_n \le 1, \quad \forall n\in\caln, \quad (\eta_n),
    \end{align}
\end{subequations}
where we relax the binary variable $x_n\in\{0,1\}$ to be continuous, i.e., $x_n \in [0,1]$, and
$\bmu:=\{\mu_n\}_{n\in\caln}$, $\bla:= \{\lambda_t\}_{t\in\calt}$, $\bsigma:=\{\sigma_{nt}\}_{n\in\caln,t\in\calt_n}$, $\bet := \{\eta_n\}_{n\in\caln}$ are the Lagrange multipliers associated with the corresponding constraints.
The dual of the problem~\eqref{prelim_relaxed_prim} is:
\begin{subequations}\label{prelim_dual}
    \begin{align}
    \min_{\bmu,\bla,\bsigma,\bet\ge0} \quad& \sum_{t\in\calt}\lambda_t C +\sum_{n\in\caln}\eta_n\\
    \label{eq:dual-const1}
    {\rm s.t.}\quad& v_n- \mu_n E_n  + \sum_{t\in\calt_n}\sigma_{nt}R_n -\eta_n \le 0, \forall n\in\caln,\\
    \label{eq:dual-const2}
    & \mu_n-\lambda_t - p_t -\sigma_{nt} \le 0,\forall n\in\caln,t\in\calt_n.
    \end{align}
\end{subequations}
The key idea of the \opd approach is to construct a feasible dual solution based on the online solution produced by an online algorithm, and then build the upper bound of the offline optimum using the feasible dual objective based on weak duality. 
Given an arrival instance $\cali$, we denote the online solution by $\hat{X}: = (\hat\bx,\hat\by)$ and the constructed dual variable by $\bar{\Lambda} := (\bar\bmu,\bar\bla,\bar\bsigma,\bar\bet)$. Also let $P_n$ and $D_n$ denote the objective values of the primal problem~\eqref{prelim_relaxed_prim} and dual problem~\eqref{prelim_dual} after the $n$-th EV is processed, respectively. Then the \opd can be summarized as the following proposition.
\begin{pro}[Proposition~3.1~\cite{tan2020mechanism}]
\label{lem:local-conditions}
An online algorithm is $\alpha$-competitive if the following conditions hold:

(i)  $\hat{X}$ and $\bar{\Lambda}$ are feasible solutions of the primal problem~\eqref{prelim_relaxed_prim} and the dual problem~\eqref{prelim_dual}, respectively;

(ii) There exists an index $k\in\caln$ such that the following sufficient inequality holds:
\begin{itemize}
    \item (initial inequality) $P_k \ge \frac{1}{\alpha}D_k,$
    \item (incremental inequality) $P_n - P_{n-1} \ge \frac{1}{\alpha} (D_n - D_{n-1}), \forall n = k+1,\dots,N.$
\end{itemize}
\end{pro}





Next we prove Lemma~\ref{eq:threshfunc1} based on Proposition~\ref{lem:local-conditions}.

\noindent\textbf{Constructing feasible solutions.} Given the feasible primal solution $\hat{X}: = (\hat\bx,\hat\by)$ produced by $\opa(\phi)$, we construct the dual solutions $\bar{\Lambda}: =(\bar\bmu,\bar\bla,\bar\bsigma,\bar\bet )$ as follows:
\begin{subequations}\label{eq:feasible-dual}
\begin{align}
\label{eq:dual-la}
\bar{\lambda}_t &= [\phi_t(w_t^{(N)}) - p_t]\cdot\mathbb{I}{\{w_t^{(N)} \ge \beta\}},\forall t\in\calt,\\
\label{eq:dual-mu}
\bar\mu_n &= \hat{\mu}_n = \max_{t\in\calt_n:\hat{y}_{nt} >0} \phi_t(w_t^{(n)}),
\\
\label{eq:dual-beta}
\bar\sigma_{nt}&= \hat{\sigma}_{nt} = [\hat{\mu}_n - \phi_t(w_t^{(n)})]\cdot\mathbb{I}{\{\hat{y}_{nt}>0\}},\\
\bar{\eta}_n &=
\begin{cases}
v_n- \bar\mu_n E_n +\sum_{t\in\calt_n}\bar\sigma_{nt}R_n, & \hat{x}_n = 1,\\
0, & \hat{x}_n = 0,
\end{cases}
\end{align}
\end{subequations}
where $\hat{\mu}_n$ and $\{\hat{\sigma}_{nt}\}_{t\in\calt_n}$ are the optimal dual variable of the cost minimization problem~\eqref{p:cost-minimization}, and $w_t^{(n)}$ is the utilization of the charging station at slot $t$ after processing EV $n$.

\begin{lemma}
The constructed dual variables in Equation~\eqref{eq:feasible-dual} are feasible for the problem~\eqref{prelim_dual}.
\end{lemma}

\begin{proof}
It can be easily checked that all constructed dual variables are non-negative.
We aim to show the constructed dual variables satisfy dual constraints~\eqref{eq:dual-const1} and~\eqref{eq:dual-const2}.
 Since the cost minimization problem~\eqref{p:cost-minimization} is convex under the capacity-free instances, by checking its KKT conditions, we have the equations as follows:
\begin{subequations}\label{eq:KKT}
\begin{align}
 \label{eq:KKT1}
    &\phi_{t}(w_{t}^{(n)}) - \hat{\mu}_n + \hat{\sigma}_{nt} - \hat{\gamma}_{nt} = 0,\quad \forall t\in\calt_n\\ \label{eq:KKT2}
    &\hat{\mu}_n E_n-\hat{\mu}_n\sum_{t \in \calt_n}\hat{y}_{nt} = 0\\ \label{eq:KKT3}
    &\hat{\sigma}_{nt}\hat{y}_{nt}-\hat{\sigma}_{nt}R_n = 0,\quad \forall t\in\calt_n \\
    \label{eq:KKT4}
    &\hat{\gamma}_{nt} \hat{y}_{nt} = 0, \forall t\in\calt_n,
\end{align}
\end{subequations}
where the last equation is from the constraint $y_{nt} \ge 0$ and $\hat{\gamma}_{nt}$ is the corresponding dual variable.

By multiplying $\hy_{nt}$ on both sides of Equation~\eqref{eq:KKT1}, summing over $\calt_n$, and substituting equations~\eqref{eq:KKT2}-\eqref{eq:KKT4}, we can finally have
\begin{align}\label{eq:relation}
    \sum_{t\in\calt_n} \hy_{nt}\phi_t(w_t^{(n)}) = \hat{\mu}_t E_n - \sum_{t\in\calt_n} \hat{\sigma}_{nt} R_n.
\end{align}
Note that the constructed dual variables $\bar{\mu}_{n}$ and $\bar{\sigma}_{nt}$ in Equations~\eqref{eq:dual-mu} and~\eqref{eq:dual-beta} are exactly the same as $\hat{\mu}_n$ and $\hat{\sigma}_{nt}$.
Then we can check the dual constraints~\eqref{eq:dual-const1}.
When ${\hat{x}_n = 1}$, we have $v_n - \bar\mu_n E_n  +\sum_{t\in\calt_n}\bar\sigma_{nt}R_n -\bar\eta_n = 0$.
When ${\hat{x}_n = 0}$, we have $\bar\eta_n = 0$ and $v_n < \sum_{t\in\calt_n} \hy_{nt}\phi_t(w_t^{(n)}) = \bar{\mu}_t E_n - \sum_{t\in\calt_n} \bar{\sigma}_{nt} R_n$. Thus, $\bar{\Lambda}$ satisfies constraint~\eqref{eq:dual-const1}.

Next, we check the dual constraint~\eqref{eq:dual-const2}.

When $\hat{y}_{nt}>0$, we have $\bar\sigma_{nt} = \bar\mu_n - \phi_t(w_t^{(n)})$. Thus, $\bar\mu_n - \bar\la_t - p_t -\bar\sigma_{nt} = \bar\mu_n - \phi_t(w_t^{(n)}) -\bar\sigma_{nt} = 0$.
When $\hat{y}_{nt}=0$, $\bar\mu_n < \phi_t(w_t^{(n)})$ and $\bar\sigma_{nt} = 0$. This gives $\bar\mu_n - \bar\la_t - p_t -\bar\sigma_{nt} = \bar\mu_n - \phi_t(w_t^{(n)}) < 0$. Thus, $\bar{\Lambda}$ satisfies constraint~\eqref{eq:dual-const2} and this completes the proof.
\end{proof}
 

\noindent\textbf{Guaranteeing the sufficient inequality:} \label{subsec:guar_incr_ineq}
We aim to find a set of pricing functions $\phi$ and the corresponding CR $\alpha$ such that the sufficient inequality holds.

We first show there exists $k\in\caln$ such that the initial inequality holds. Let $k$ be the index of the first EV, after processing which the station utilization reaches $\beta$ in at least one slot. Under the infinitesimal scheduling assumption~\ref{assumption1}, we can consider there exist a set of time slots $\bar{\calt}_k \subseteq \calt_k$ such that $w_{t}^{(k)} = \beta, \forall t\in \bar{\calt}_k$ and $w_{t}^{(k)} < \beta, \forall t\in \calt \setminus \bar{\calt}_k$. Note that the first $k$ EVs are all admitted since the posted price $\xi_{n} = \sum_{t\in\calt_n} \hy_{nt}\phi_t(w_t^{(n)}) = L E_n \le v_n, \forall n\in[k]$.
Then we can have
\begin{subequations}
  \begin{align}
    P_k &= \sum_{n\in[k]} v_n - \sum_{n\in[k]}\sum_{t\in\calt_n} p_t \hat{y}_{nt},\\
    \label{eq:in1}
    &= \sum_{n\in[k]} [\bar{\eta}_n + \sum_{t\in\calt_n} \hy_{nt}\phi_t(w_t^{(n)})] - \sum_{n\in[k]}\sum_{t\in\calt_n} p_t \hat{y}_{nt}\\
    \label{eq:in2}
    &= \sum_{n\in[k]}\bar{\eta}_n  + \sum_{t\in\calt}  (L - p_t)\sum_{n\in[k]} \hy_{nt} \\
    &\ge \sum_{n\in[k]}\bar{\eta}_n  + \sum_{t\in\bar{\calt}_k}  (L - p_t)\beta\\
    \label{eq:in3}
    &\ge \frac{1}{\alpha}\sum_{n\in[k]}\bar{\eta}_n  + \sum_{t\in\bar{\calt}_k}  (L - p_t)\frac{C}{\alpha} = \frac{1}{\alpha} D_k,
\end{align}  
\end{subequations}
where \eqref{eq:in1} is obtained by substituting the definition of $\bar{\eta}$ and  Equation~\eqref{eq:relation}, \eqref{eq:in2} holds since $w_t^{(k)}\le \beta, \forall t\in\calt$, and \eqref{eq:in3} holds since $\beta \ge C/\alpha$.
Thus, given the sufficient condition in Lemma~\ref{eq:threshfunc1}, the initial inequality holds.


The increments of primal and dual objectives can be cast as
\begin{align*}
  P_n - P_{n-1} &=  v_{n} -  \sum_{t \in \calt_n}{p_{t} \hy_{nt}}, \\
  D_n - D_{n-1} &= \sum_{t \in \calt_n} {C(\phi_{t}(w_{t}^{(n)}) - \phi_{t}(w_{t}^{(n-1)}))} + \bar{\eta},
\end{align*} 
where we use $\phi_t(w_t^{(N)}) = \sum_{n \in \caln} [\phi_t(w_t^{(n)}) - \phi_t(w_t^{(n-1)})]$.
When $\hx_n = 0$, we have $P_n - P_{n-1} =  D_n - D_{n-1} = 0$ and thus the incremental inequality holds.  
When $\hx_n = 1$, we can show that the incremental inequality holds if the pricing function and the competitive ratio $\alpha$ satisfy the sufficient condition~\eqref{eq:thresh_ineq1} in Lemma~\ref{eq:threshfunc1}. 
To see this, note
\begin{subequations}
\begin{align}
 &D_n - D_{n-1}  \notag\\
 &=\sum_{t \in \calt_n} {C(\phi_{t}(w_{t}^{(n)}) - \phi_{t}(w_{t}^{(n-1)}))} + v_{n} - \bar{\mu}_{n}E_{n} +\sum_{t\in\calt_{n}}{\bar{\sigma}_{nt}R_{n}}\notag\\
 \label{eq:ineq0}
 &= \sum_{t \in \calt_n} {C(\phi_{t}(w_{t}^{(n)}) - \phi_{t}(w_{t}^{(n-1)}))} + v_{n} - \sum_{t\in\calt_n} \hy_{nt}\phi_t(w_t^{(n)})\\
 &= \sum_{t \in \calt_n} \left[{C(\phi_{t}(w_{t}^{(n)}) - \phi_{t}(w_{t}^{(n-1)})) - \hy_{nt}\phi_t(w_t^{(n)})}
 \right] + v_n\\ 
 \label{eq:ineq1}
 &\le \sum_{t \in \calt_n} \hy_{nt}\left[{C\phi_t'(w_t^{(n)}) - \phi_t(w_t^{(n)})}
 \right] + v_n\\
 \label{eq:ineq2}
 &\le \sum_{t \in \calt_n} \hy_{nt} \left[ (\alpha-1)\phi_t(w_t^{(n)}) - \alpha p_t\right] + v_n\\
 \label{eq:ineq3}
 &\le (\alpha - 1) v_n -  \alpha \sum_{t \in \calt_n} \hy_{nt}p_t + v_n = \alpha (P_n - P_{n-1}),
\end{align}    
\end{subequations}
where we substitute equation~\eqref{eq:relation} to obtain~\eqref{eq:ineq0}, use the inequality $\phi_{t}(w_{t}^{(n)}) - \phi_{t}(w_{t}^{(n-1)}) \le \hy_{nt}\phi_t'(w_t^{(n)})$ based on assumption~\ref{assumption1} in~\eqref{eq:ineq1}, apply the sufficient condition~\eqref{eq:thresh_ineq1} in~\eqref{eq:ineq2}, and use inequality $v_n \ge \sum_{t\in\calt_n} \hy_{nt} \phi_t(w_t^{(n)})$ (since $\hat{x}_n = 1$ in $\opa(\phi)$) in~\eqref{eq:ineq3}.

Therefore, $\opa(\phi)$ is $\alpha$-competitive if the sufficient condition~\eqref{eq:thresh_ineq1} holds in Lemma~\ref{eq:threshfunc1}. This completes the proof.  

\subsection{Proof of Lemma~\ref{lem:cap_lim_comp_ratio}} \label{proof:cap_lim_comp_ratio}
Let $\cali$ be an instance under the capacity-limited cases $\Psi^2$ and $\{w_t^{(N)}\}_{t \in \calt}$ be the final utilization of all time slots after \texttt{OPA} runs for all EVs in $\cali$. Assume the time horizon $T$ is long and is an integer multiple of $D^{\max}$. We divide the time horizon into $H = T/ D^{\max}$ partitions. We define $\calt^{h} := \{t \in  \calt: (h-1)D^{\max} + 1 \le t \le h D^{\max}\}$ as the set of time slots in the $h$-th time partition, and $\hat{\calt}^h = \calt ^ h \cup \calt ^{h+1}$. 
$\cali^{h}$ is the sub-instance of $\cali$ which includes EVs whose arrival times are in $\calt^{h}$. Furthermore, we define $\tilde{\cali}^h$ as a 3-partition sub-instance $\tilde{\cali}^h = \cali ^ {h-1} \cup \cali ^ {h} \cup \cali ^ {h+1}, h \in [H]$, where $\cali ^ 0 = \cali ^ {H+1} = 0$. 
The competitive ratio of \texttt{OPA} under instance $\cali$ is:

\begin{align}
    &\CR(\phi^*) = \frac{\opt(\cali)}{\alg(\cali)} = \frac{\sum_{h\in[H]}\opt(\cali^h)}{\sum_{h\in[H]}\alg(\cali^h)}\label{eq:upbound_comp_ratio_caplim}\\
    &= \frac{3\sum_{h\in[H]}\opt(\cali^h)}{\alg({\cali}^1)+\sum_{h\in[H]}\alg(\tilde{\cali}^h) + \alg({\cali}^H)} \le 3 \max_{h\in[H]} \frac{\opt(\cali^h)}{\alg(\tilde{\cali}^h)}\notag.
\end{align}
Therefore, for calculating CR, we need to calculate $\max_{h\in[H]} \frac{\opt(\cali^h)}{\alg(\tilde{\cali}^h)}$.

\begin{pro} \label{prepos1}
The total surplus of EVs that arrive in $\tilde{\cali}^h$ and accept the posted prices from $\opa(\phi^*)$ is lower bounded by
\end{pro}

\begin{equation} \label{c44}
    \alg(\tilde{\cali}^h) \ge \frac{1}{\alpha} \sum_{t \in \hat{\calt}^{h}} \phi_t(w_t^{(N)})C.
\end{equation}

\begin{proof}
We have
\begin{align} \label{c22}
    \sum_{t \in \hat{\calt}^{h}} \phi_t(w_t^{(N)})C =&\sum_{t \in \hat{\calt}^{h}} \sum_{n \in \cali} C[\phi_t(w_t^{(n)}) - \phi_t(w_t^{(n-1)})]\\
    =& \sum_{t \in \hat{\calt}^{h}} \sum_{n \in \tilde{\cali}^h: t \in \calt_n} C[\phi_t(w_t^{(n)}) - \phi_t(w_t^{(n-1)})]\notag\\
    \le &\sum_{n \in \tilde{\cali}^h} \sum_{t \in \calt_n} C[\phi_t(w_t^{(n)}) - \phi_t(w_t^{(n-1)})],\notag
\end{align}
where $\calt_n$ is the availability window of EV $n$. The second equality holds because the maximum duration of each EV is $D^{\max}$; hence, the EVs that can stay in $\hat{\calt}^h$ must be from $\tilde{\cali}^h$. The last inequality holds since the EVs in $\tilde{\cali}^h$ can stay up to partition $h + 2$.

We define $\Delta \alg_n = (v_n - \sum_{t \in \calt_n} p_t y_{nt})\hat{x}_n$ as the surplus increment of $\opa(\phi^*)$ after EV $n$ is processed. 
Next we aim to show
\begin{align} \label{ineq:case_1_2}
    \sum_{t \in \calt_n} C[\phi_t(w_t^{(n)}) - \phi_t(w_t^{(n-1)})] \le \alpha \Delta \alg_n.   
\end{align}


\noindent\textbf{Case(i).} When EV $n$ rejects the offer, both sides of~\eqref{ineq:case_1_2} are zero; thus, it stands correct.

\noindent\textbf{Case(ii).} When EV $n$ is admitted, we have $w_t^{(n)} = w_t^{(n-1)} + \hat{y}_{nt}$. Because the pricing function~\eqref{eq:thresh_func} has three segments; thus, we have three scenarios based on which segment $w_t^{(n)}$ and $w_t^{(n-1)}$ lie in.\\

\noindent\textbf{Case(iia)}: when $w_t^{(n)} \ge \beta$ and $w_t^{(n-1)} \ge \beta$, we have
\begin{subequations} \label{c33}
    \begin{align}
        &\sum_{t \in \calt_n} C[\phi_t(w_t^{(n)}) - \phi_t(w_t^{(n-1)})] \\
        &= C\sum_{t \in \calt_n} G \exp(w_t^{(n)} \alpha/C)[1- \exp(- \hat{y}_{nt} \alpha/C)]\label{c1}\\
        &\le C\sum_{t \in \calt_n} G \exp(w_t^{(n)} \alpha/C)\hat{y}_{nt} \alpha/C \label{c2}\\
        &= \alpha \sum_{t \in \calt_n} \hat{y}_{nt} [G \exp(w_t^{(n)} \alpha/C) + p_t] - \alpha \sum_{t \in \calt_n} p_t \hat{y}_{nt}\label{c3}\\
        &= \alpha[ \sum_{t \in \calt_n} \hat{y}_{nt} \phi_t(w_t^{(n)}) - \sum_{t \in \calt_n} p_t \hat{y}_{nt}]\label{c4}\\
        &\le \alpha [v_n - \sum_{t \in \calt_n} p_t \hat{y}_{nt}]] 
        \le \alpha \Delta \alg_n, \label{c6}
    \end{align}
\end{subequations}
where $G := \frac{L-p_t}{e}$, Equation~\eqref{c1} is obtained by substituting pricing function to~\eqref{eq:thresh_func}, and Inequality~\eqref{c6} holds because when EV $n$ admits the offer, we have $v_n \ge \sum_{t \in \calt_n} \hat{y}_{nt} \phi_t(w_t^{(n-1)})$ .

\noindent\textbf{Case(iib)}: when $w_t^{(n)} > \beta$ and $w_t^{(n-1)} < \beta$, we have
\begin{subequations}
    \begin{align}
        &\sum_{t \in \calt_n} C[\phi_t(w_t^{(n)}) - \phi_t(w_t^{(n-1)})] \\
        &= C\sum_{t \in \calt_n} [\phi_t(w_t^{(n)}) - L]\\
        \label{eq:s1}
        &\le C\sum_{t \in \calt_n} [\phi_t(w_t^{(n)}) - (G \exp(w_t^{(n-1)} \alpha/C) + p_t)]\\
        &= C\sum_{t \in \calt_n} G \exp(w_t^{(n)} \alpha/C)[1- \exp(-\hat{y}_{nt} \alpha/C)],
    \end{align}
\end{subequations}
where Inequality~\eqref{eq:s1} holds by noting that $G \exp(w_t^{(n-1)} \alpha/C) + p_t \le \phi_t(w_t^{(n-1)}) = L$. Then, we can continue the proof in the same procedure from Inequality~\eqref{c2} to ~\eqref{c6}.

\noindent\textbf{Case(iic)}: when $w_t^{(n)} < \beta$ and $ w_t^{(n-1)} < \beta$, we have $\phi_t(w_t^{(n)})=\phi_t(w_t^{(n)})=L$ and thus $\sum_{t \in \calt_n} C[\phi_t(w_t^{(n)}) - \phi_t(w_t^{(n-1)})] = 0 \le \alpha \Delta \alg_n$.
Therefore, in all cases, Inequality~\eqref{ineq:case_1_2} stands correct, and by combining
inequalities~\eqref{c22} and~\eqref{ineq:case_1_2}, we have Inequality~\eqref{c44}.
\end{proof}

To calculate the upper bound of the competitive ratio under capacity-limited instances $\Psi^2$, we aim to lower bound the objective value of $\opa(\phi^*)$ and upper bound offline optimum. 
In Figure~\ref{fig:worst_case}, we depict an instance that can lead to the least possible utilization in $\Psi^2$. 
In this instance, there are two groups of EVs that follow three rules: (a) group 1 leaves the charging station at time slot $t'$ and group 2 arrives at the charging station at time slot $t'$; (b) each EV in each group is available at the charging station for $D^{\min}$ slots; and (c) the energy demand of each EV equals $D^{\min}R_n$, which means each EV must be charged up to its rate limit in each time slot during its available window. 
We define $k_1$ ($k_2$) as the final utilization from group 1 (group 2) and $k_1$ and $k_2$ can be visualized by the heights of the red and blue rectangles in Figure~\ref{fig:worst_case}.

{Before we proceed, we note that the scheduling part of the proposed algorithm as a solution to problem~\eqref{p:cost-minimization} is a water-filling solution, which can be shown as
\begin{align}\label{eq:water-filling}
    \hat{y}_{nt} = \min\left\{\left[\phi^{-1}_t(\hat{\mu}_n) - w_t^{(n-1)}\right]^+, R_n\right\},
\end{align}
where $\hat{\mu}_n$ can be determined by solving $\sum_{t\in\calt_n} \hat{y}_{nt} = E_n$.
Thus, the per-slot charging rate scheduling starts with the least utilized slot and fills it up to either the water level or charging rate. The same process then will be applied to the next least utilized slot until the charging demand is satisfied.
This water-filling solution is derived based on the KKT condition~\eqref{eq:KKT}. Consider the following three cases: (i) if $\phi_t(w_t^{(n)}) > \hat{\mu}_n$, we have $\hat{\gamma}_{nt} > \hat{\sigma}_{nt} \ge 0$ and this gives $\hat{y}_{nt} = 0$ from the complementary slackness; (ii) if $\phi_t(w_t^{(n)}) < \hat{\mu}_n$, we have $\hat{\sigma}_{nt} > \hat{\gamma}_{nt} \ge 0$, and thus $\hat{y}_{nt} = R_n$; (iii) if $\phi_t(w_t^{(n)}) = \hat{\mu}_n$, we have $\hat{y}_{nt} = \phi^{-1}_t(\hat{\mu}_n) - w_t^{(n-1)}$. Combining these three cases gives the water-filling solution~\eqref{eq:water-filling}, and the determination of $\hat{\mu}_n$ is based on Equation~\eqref{eq:KKT2}.
}


\begin{figure}[t]
\includegraphics[width=0.6\linewidth, center]{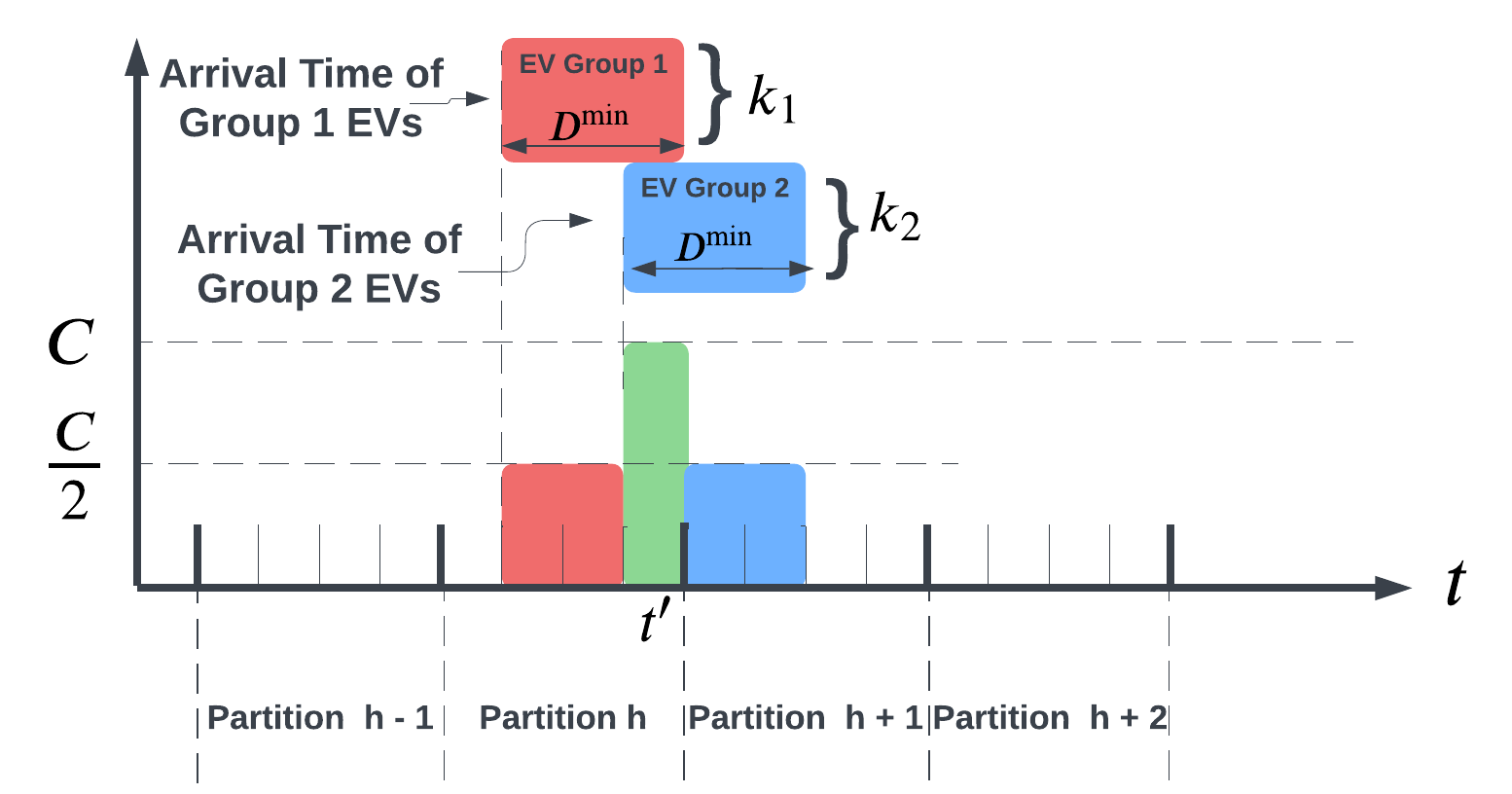} 
\caption{Illustration of the worst-case final utilization.} 
\label{fig:worst_case}
\end{figure}

In the following, we can argue that the final utilization obtained by $\opa(\phi^*)$ under any capacity-limited instance is no less than that depicted in Figure~\ref{fig:worst_case}.
First, there exists at least one time slot ($t'$ in Figure~\ref{fig:worst_case}) whose utilization is at its full capacity based on the definition of the capacity-limited instance.
To ensure the utilization of the overlapping slot reaches the capacity, each EV cannot have flexibility to shift its charging demand over time (each EV has $D_{\min}$ slots of available window and $r_n D_{\min}$ energy demand in Figure~\ref{fig:worst_case}) since otherwise the water-filling solution leads to higher utilization for the slots excluding $t'$.

We define $\psi(w)$ as
\begin{align*}
    \psi(w) = \begin{cases}
    L, & w \in [0, C/\alpha)\\
    \frac{L - p^{\max}}{e} e^{\frac{\alpha}{C} w} + p^{\max}, & w \in [C/\alpha, C]
    \end{cases},
\end{align*}
and it can be easily verified that
\begin{align}\label{eq:min_price_func}
   \psi(w) \le \phi_t^*(w), \quad \forall t \in \hat{\calt}^{h}, w \in [0,C]. 
\end{align}
Based on Proposition~\ref{prepos1} and Inequality~\eqref{eq:min_price_func}, we have
\begin{align}\label{eq:online_lower_bound}
    \texttt{ALG}(\tilde{\cali}^h) \ge \frac{C}{\alpha} [(D^{\min} - 1)\psi(k_1) + (D^{\min}-1) \psi(k_2) + \psi(C)].
\end{align}

According to Jensen's inequality and convexity of the pricing function, we additionally have $\psi(k_1) + \psi(k_2) \ge \psi(\frac{C}{2}) + \psi(\frac{C}{2})$ when $k_1 + k_2 = C$.
Because our pricing function ~\eqref{eq:thresh_func} is a two segment function, we have two cases:

\noindent\textbf{Case(i)}: if $\alpha \le 2$, then we have $k_1 = k_2 = \frac{C}{2} \le \beta = \frac{C}{\alpha}$ and
\begin{subequations}\label{eq:c1}
    \begin{align}
        \texttt{ALG}(\tilde{\cali}^h) \ge &\frac{C}{\alpha} [(2D^{\min}-2)L + \psi(C) ]\label{c1_subeq2}\\
        \ge &\frac{2D^{\min}LC}{\alpha}\label{c2_subeq2}\\ 
        \ge& \frac{2D^{\min}C(L-p^{\max})}{\alpha}\label{c1_subeq3},
    \end{align}
\end{subequations}
where Inequality~\eqref{c1_subeq2} stands correct because $\psi(k_1) = \psi(k_2) = L$.
Inequality~\eqref{c2_subeq2} is correct since ${(U/L)(D^{\max}/D^{\min}) \ge 2}$, and 
$\psi(C) \ge UD^{\max}/D^{\min}$ so we have: $ \psi(C) \ge 2L$.\\

\indent\textbf{Case(ii)}: if $\alpha > 2$, then $ k_1 = k_2 = \frac{C}{2} > \beta = \frac{C}{\alpha}$, and 
\begin{subequations}\label{eq:c2}
    \begin{align}
    \label{c2_subeq0}
        \texttt{ALG}(\tilde{\cali}^h) \ge &\frac{C}{\alpha} \Bigg((2D^{\min}-2)(\frac{L-p^{\max}}{e}) \exp (\frac{\alpha}{2})  +(\frac{L-p^{\max}}{e}) \exp (\alpha)\Bigg) \\
        \ge &\frac{C}{\alpha} \Bigg(2D^{\min} (\frac{L-p^{\max}}{e}) \exp(\frac{\alpha}{2})\Bigg), \label{c2_subeq1}
    \end{align}    
\end{subequations}

Inequality~\eqref{c2_subeq1} also stands because when $\alpha \ge 2$, we have $$(\frac{L-p^{\max}}{e}) \exp (\alpha) > 2 (\frac{L-p^{\max}}{e}) \exp (\frac{\alpha}{2}).$$

The upper bound of offline optimum is $\opt(\cali^h) \le 2D^{\max}CU$, which is the maximum value the offline algorithm could achieve in $2D^{\max}$ time slots, and the competitive ratio is as follows.
When $\alpha \le 2$, according to Inequality~\eqref{eq:c1}, we have
\begin{align*}
    \frac{\texttt{OPT}(\cali^h)}{\texttt{ALG}(\tilde{\cali}^h)} \le \frac{D^{\max} U  \alpha}{ D^{\min}(L-p^{\max})} \le 2\theta = 2\sqrt{e},
\end{align*}
where we apply $\alpha = 1 + 2\ln(\theta) = 2$.

When $\alpha > 2$, according to Inequality~\eqref{eq:c2} we have
\begin{align*}
    \frac{\texttt{OPT}(\cali^h)}{\texttt{ALG}(\tilde{\cali}^h)} \le \frac{D^{\max} U \alpha}{D^{\min} (\frac{L-p^{\max}}{e})  \exp (\frac{\alpha}{2})} = \frac{\theta \alpha}{\exp(\alpha/2 - 1)}.
\end{align*}

Combining the above two cases gives
\begin{align*}
    \frac{\texttt{OPT}(\cali^h)}{\texttt{ALG}(\tilde{\cali}^h)} \le \max\left\{2\sqrt{e}, \frac{\theta\alpha}{\exp(\alpha/2 - 1)}\right\}.
\end{align*}

Therefore, according to Inequality ~\eqref{eq:upbound_comp_ratio_caplim}, we have
\begin{align*}
    \frac{\opt(\cali)}{\alg(\cali)} &\le 3 \max_{h\in[H]}\frac{\opt(\cali^h)}{\alg(\tilde{\cali}^h)}\\
    &\le 3\max\left\{2\sqrt{e}, \frac{\theta\alpha}{\exp(\alpha/2 - 1)}\right\}.
\end{align*}
This completes the proof.

\section{Experimental Results}\label{sec:exp_res}
In this section, we report the experimental results using EV charging data traces. The goal is to evaluate the performance of the proposed algorithms as compared to alternatives and investigate the impact of the parameters on the performance of the algorithms. 
\subsection{Experimental Setup}
\noindent\textbf{Dataset.}
We use the Caltech ACN dataset~\cite{lee2019acn}, which includes more than 50K EV charging sessions from more than 50 charging stations. The dataset includes the arrival and departure time, the charging demand, and the charging rate for each EV. 

\noindent\textbf{Parameter settings and metrics.} To capture the scenario of limited capacity, we consider a charging station with an aggregate capacity smaller than the total charging demand of EVs. We conduct a 1-day trial using 90 consecutive days in which we generate 20 trials randomly and calculate the average of the 20 profit ratios for each day. Each EV $n$ has a different valuation ($v_n$) in each trial in a day, making $90 \times 20 = 1800$ trials in total. The ACN dataset does not include the private value ($v_n$) for EVs; hence, we follow the approach in~\cite{zeynali2021data} to estimate the value for EVs by modeling the distribution of historical arrivals. Each EV's value is a randomized function of the availability window, which means the higher the window, the higher the EV's value.
As the counterpart of the competitive ratio in the empirical setting, we use \textit{empirical profit ratio} as the performance metric, which is the ratio of the profit gained by the optimal solution to the offline problem.  

\noindent\textbf{Comparison algorithms.}
We compare our algorithm \texttt{OPA} with three other online algorithms: (1) \textit{A Utility-Based Online Algorithm} (\texttt{UBOA}), where the admission control is simply based on a first-come-first-served policy. For the scheduling, \texttt{UBOA} adapts a water-filling policy and starts with the slots with lower utilization up to either the capacity or EV charging rate. 
(2) \textit{Price Based Online Algorithm} (\texttt{PBOA}), where the admission policy is first-come first-served, but the scheduling is a water-filling approach based on the time-varying unit energy cost of each time slot, i.e., the scheduler starts with the time slot with lowest unit energy cost.
(iii) \textit{Online Mechanism with Myopic Price} (\texttt{OMMP})~\cite{sun2018eliciting}: The online algorithm proposed in~\cite{sun2018eliciting}, which calculates the charging schedule by solving a cost minimization based on the per-slot energy price. This per-slot energy price is a quadratic pricing function that we are utilizing to evaluate the algorithm's performance. The original maximization problem of ours \eqref{p:social-welfare} and~\cite{sun2018eliciting} are quite similar; therefore, we modify the cost function in~\cite{sun2018eliciting} into a linear one to compare the profit ratio fairly. Therefore, \texttt{OMMP} is an online algorithm that uses a pricing function based on the linear cost function to solve a cost minimization problem to calculate the candidate charging schedule and the corresponding marginal price for each unit of energy according to the pricing function. Then \texttt{OMMP} decides to admit EV $n$ if its surplus (value - total scheduled energy price) is positive.

\subsection{Experimental Results}
\noindent\textbf{Comparison results.}
In this experiment, we compare the performance of different algorithms in three different congestion levels: \textit{low, medium, and high}, where the capacity in the station is approximately up to 60\%, 30\%, and 15\% of the aggregate EV demand, respectively.
Figure~\ref{fig:congestion} depicts the cumulative distribution function (CDF) for empirical profit ratios of 1800 trials of experiments for all four algorithms in low, medium, and high congestion. The results show that in all congestion settings, \texttt{OPA} outperforms the alternatives, and the lower the congestion is, the better the \texttt{OPA} result gets as compared to other algorithms. A notable observation is that nearly 90\%, 80\%, and 70\% of the \texttt{OPA}'s profit ratios are less than 2 in low, medium, and high congestion scenarios, while these are substantially lower for all three alternative algorithms.

\begin{figure*}[h]
    \centering
    \subfigure[High congestion]{\includegraphics[width=0.3\textwidth]{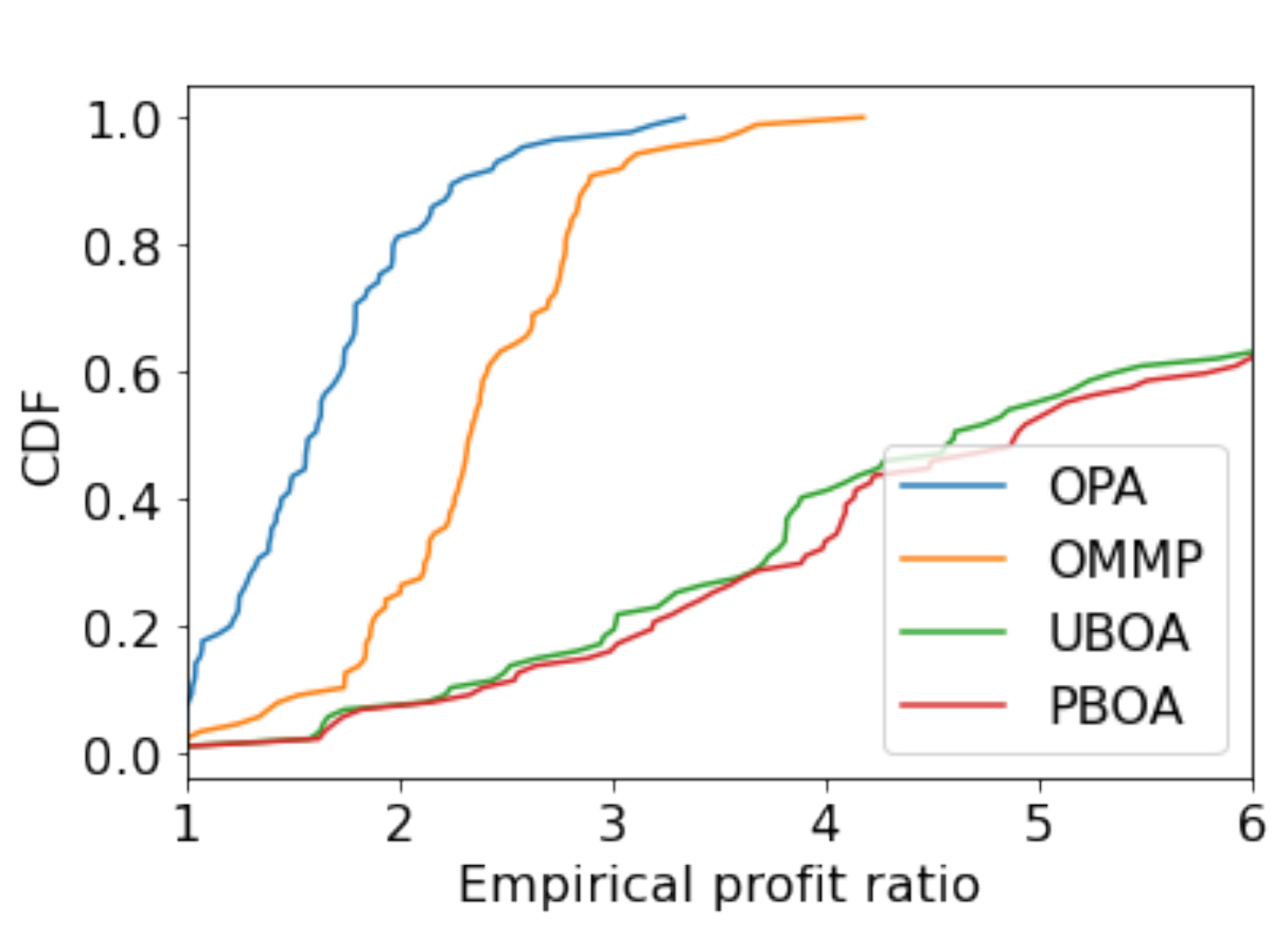}} 
    \subfigure[Medium congestion]{\includegraphics[width=0.3\textwidth]{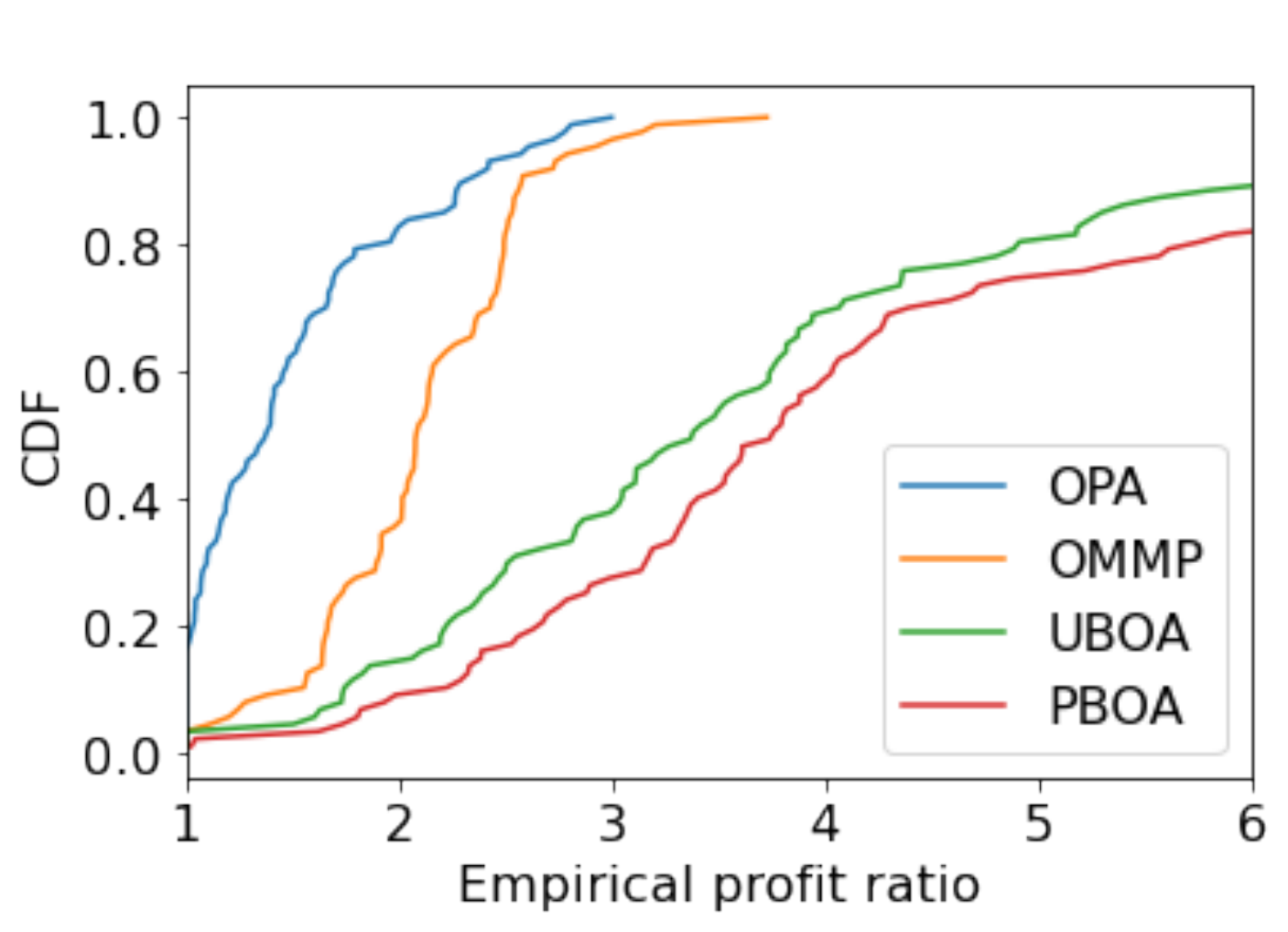}} 
    \subfigure[Low congestion]{\includegraphics[width=0.3\textwidth]{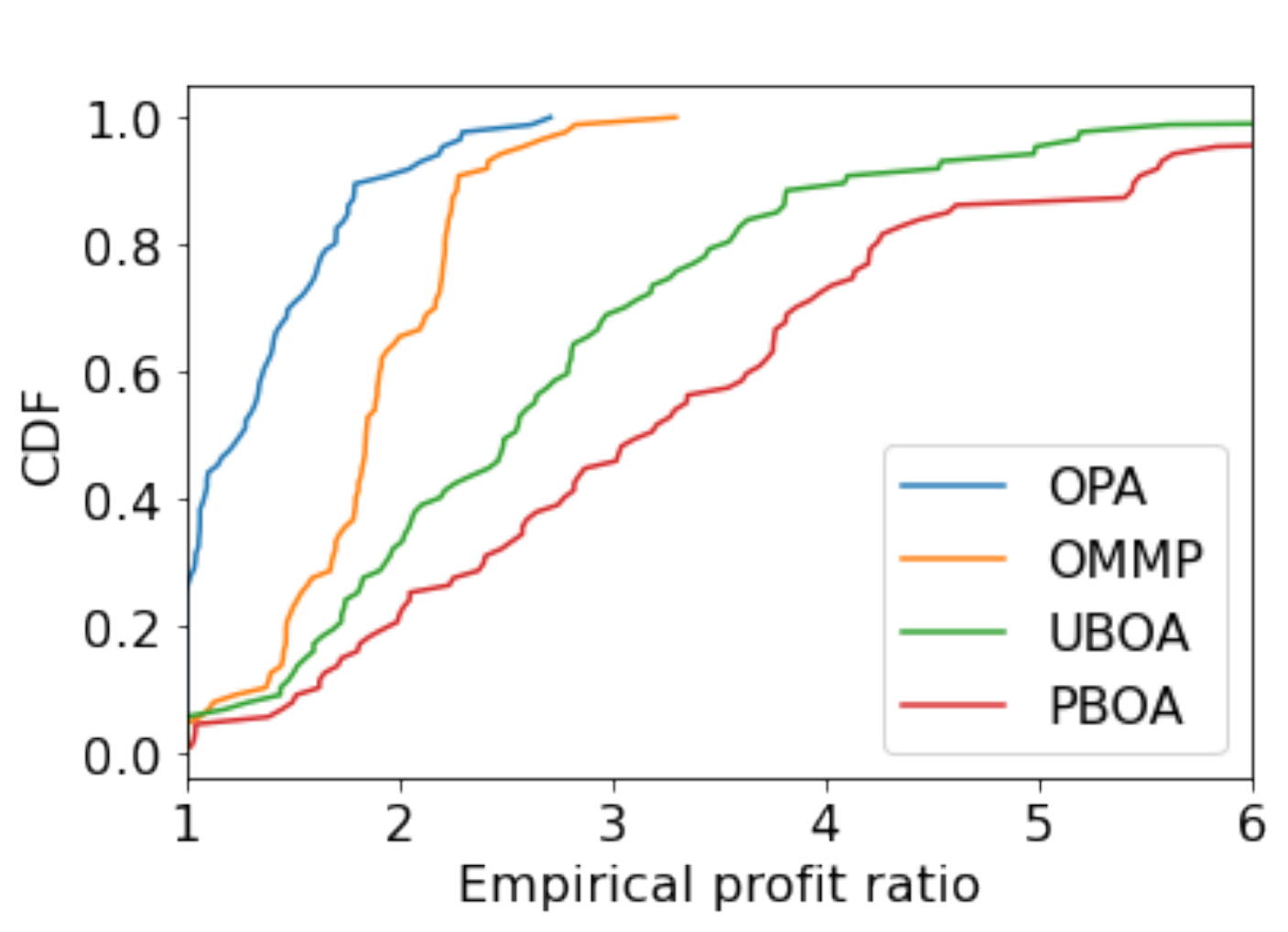}}
    \caption{Variation of the profit ratio in different congestion scenarios}
    \label{fig:congestion}
\end{figure*}

 \begin{figure*}
        \centering
         \subfigure[Capacity variation] 
            {
                \label{subfig:cap_variation}
                \includegraphics[width=.22\textwidth]{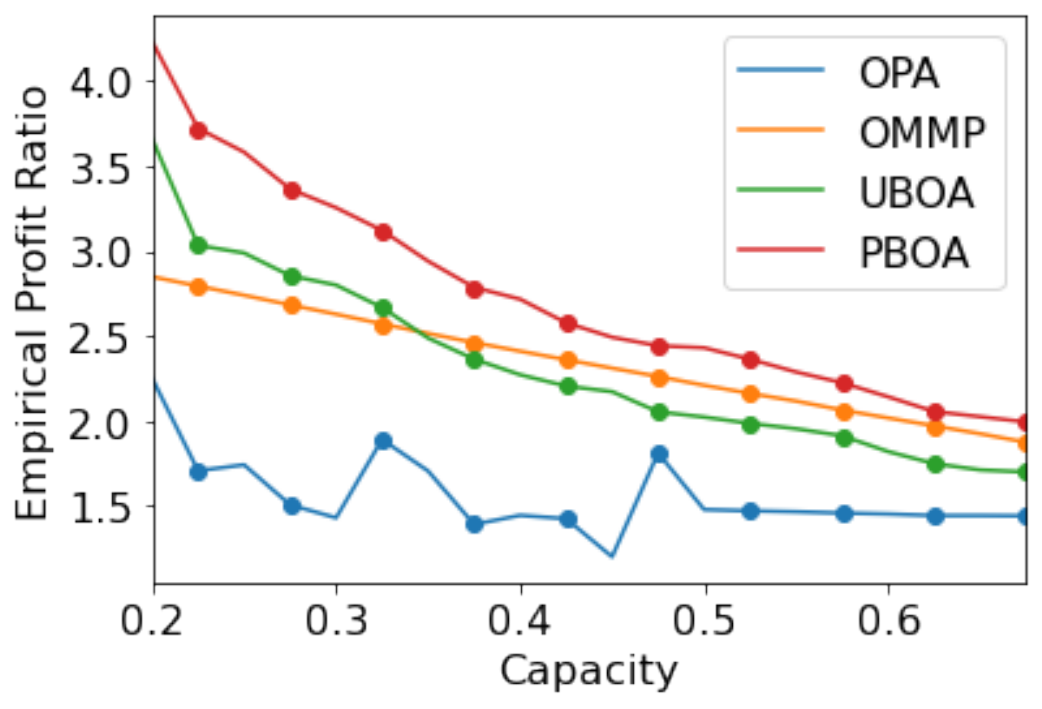} 
            } 
            \subfigure[Fluctuation ratio] 
            {
                \label{subfig:fluctuation}
                \includegraphics[width=.22\textwidth]{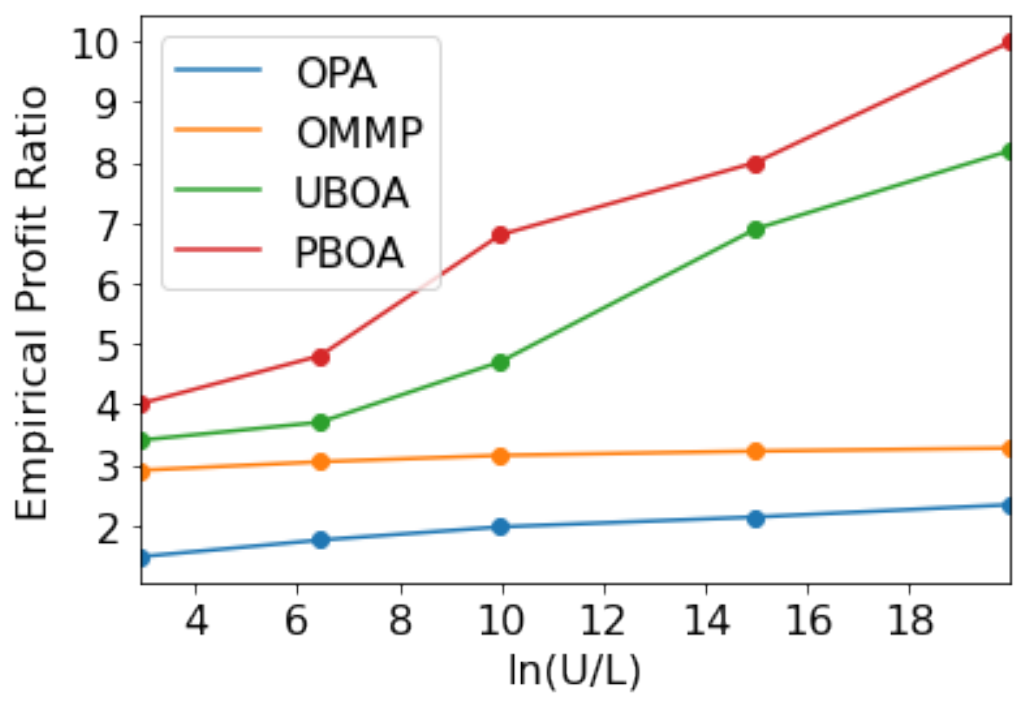} 
            }
            \subfigure[Duration variation] 
            {
                \label{subfig:duration}
                \includegraphics[width=.22\textwidth]{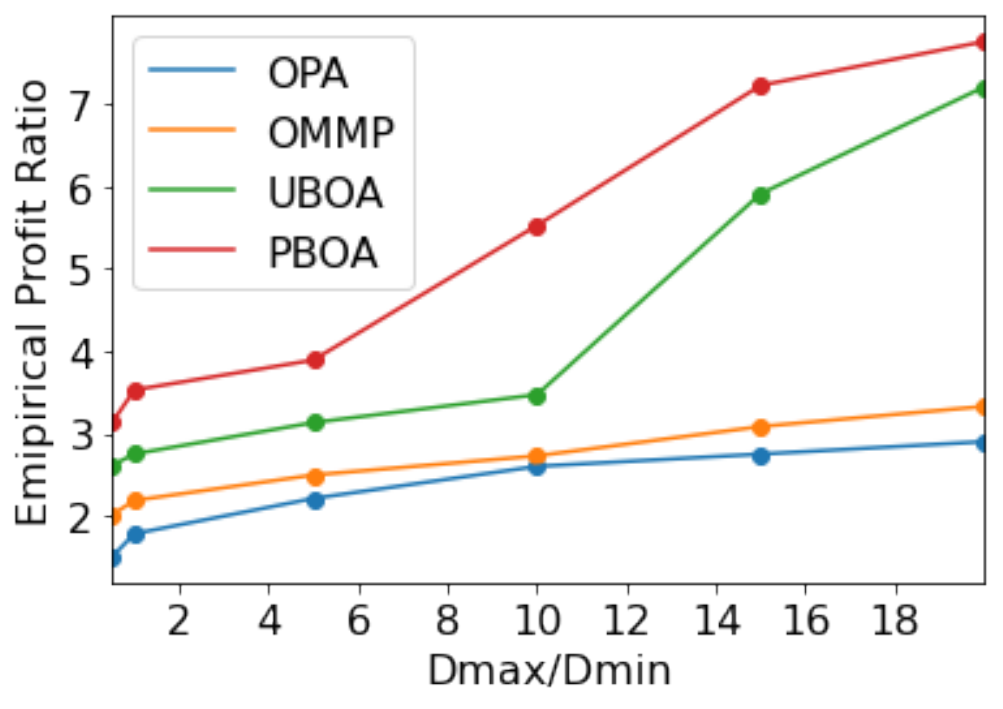} 
            } 
            \subfigure[$p^{\max}$ variation] 
            {
                \label{subfig:elec_price}
                \includegraphics[width=.22\textwidth]{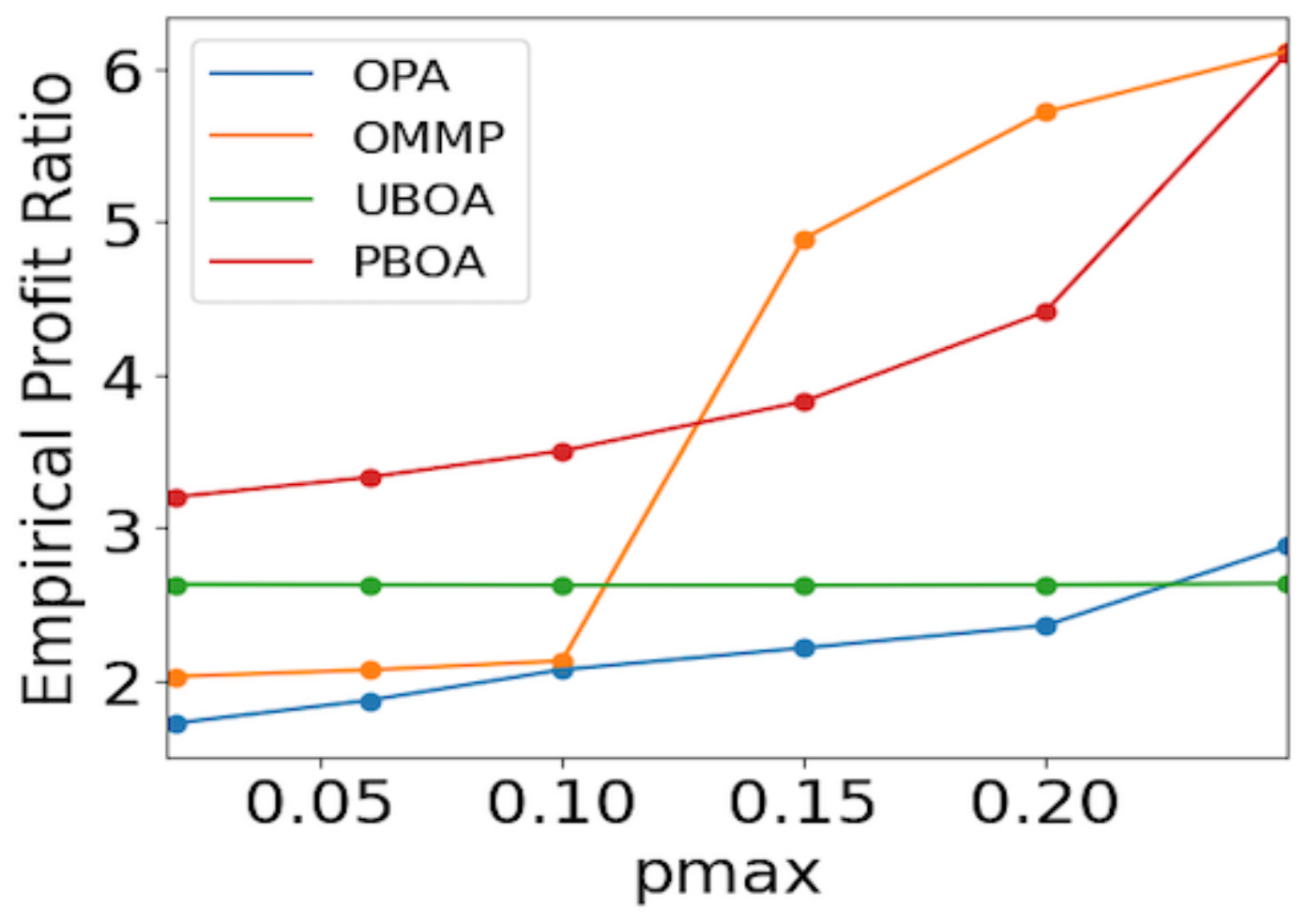} 
            }
            %
        \caption{Impact of parameters}
        \label{fig:metrics}
\end{figure*}

\noindent\textbf{The impact of parameters.}
In this experiment, we evaluate the impact of variation of the parameters on the empirical profit ratios. In particular, we vary the charging station capacity, fluctuation ratio ($\rho = U/L$), duration ratio ($\delta = D^{\max}/D^{\min}$), and the maximum electricity price ($p^{\max}$).
Figure~\ref{subfig:cap_variation} demonstrates that the empirical profit ratio decreased as the capacity increased. This observation makes sense since with more capacity, there is more room for EV admission and is aligned with the observation in three different congestion regimes in Figure~\ref{fig:congestion}. Figure~\ref{subfig:fluctuation} shows that the fluctuation ratio has a direct relation to the profit ratio, and the higher the fluctuation ratio, the higher the competitive ratio, which also approves the fact that in~\texttt{OPA} and ~\texttt{OMMP} the theoretical competitive ratio increases logarithmically with $\rho$. Figure~\ref{subfig:duration} shows when $\delta$ increases, there is a slight increase in the profit ratios of \texttt{OPA} and~\texttt{OMMP}; however, profit ratios of~\texttt{UBOA} and~\texttt{PBOA} increase drastically. 
This observation is due to the fact that different from ~\texttt{UBOA} and~\texttt{PBOA}, both \texttt{OPA} and~\texttt{OMMP} take into account the duration of each EV request in their algorithmic decision making. Finally, Figure~\ref{subfig:elec_price} shows that as the $p^{\max}$ increases, there's a slight increase in the profit ratio of \texttt{OPA}, but the slope of the increase is much smaller than those of~\texttt{OMMP} and ~\texttt{PBOA}. 
The profit ratio of~\texttt{UBOA} does not change because it is independent of electricity prices.


\section{Related Work}
\label{sec:related_works}
EV charging management has become a central topic of transportation electrification in recent decades. 
The problem is to decide where to charge, when to charge, and what rate to charge given the profile of EV requests (e.g., energy demand, departure time, etc.). 
There exist a large body of literature that schedule the EV charging in offline settings~\cite{8710609,8107290}.
This paper focuses on the online settings where EVs arrive one by one and one must decide the scheduling of each EV upon its arrival without knowing the future information.
In online EV charging, one stream of works studies how to charge EVs under unknown future electricity prices~\cite{deng2016whether, tsui2012demand, gan2012optimal, bitar2016deadline, deng2017whether, rotering2010optimal, lin2021minimizing}. In~\cite{deng2016whether}, the goal is to minimize the EV charging cost considering real-time pricing. Furthermore, in~\cite{tsui2012demand, gan2012optimal, bitar2016deadline, deng2017whether, rotering2010optimal}, the focus is on EV charging scheduling under forecasting the real-time prices or based on stochastic information of the real-time prices. The authors in~\cite{lin2021minimizing} study a multi-objective optimization to best balance the cost of charging and dissatisfaction of the EVs.
Different from this steam, we focus on the EV charging for a group of online arriving EVs from a theoretical approach of competitive analysis.

The online EV charging problem allocates limited charging capacity to a group of EVs under capacity limit where it is impossible to admit all requests and satisfy their requests. Therefore, there exists two strategies for designing the online algorithms. 
The first stream provides no guarantees on the received energy of EVs and makes best efforts to maximize the social welfare (utility of all EVs minus charging cost)~\cite{zheng2014online,lee2021adaptive,sun2020competitive}. Although~\cite{zheng2014online} works on maximizing social welfare, they do not guarantee to fully charge each admitted EV and charging station's charging capacity constraint into account. The work in~\cite{lee2021adaptive} has the charging station's charging capacity constraint. But, it does not also guarantee to charge the admitted EV up to its charging demand. In~\cite{sun2020competitive}, they propose an algorithm to maximize the social welfare considering the charging capacity. However, in~\cite{sun2020competitive}, energy cost is not considered and they only try to maximize the social welfare regardless of trying to minimize the cost EVs pay to the charging station. Our work is also more challenging than ~\cite{sun2020competitive} because we guarantee that if an EV is admitted, it will be fully charged to its demand; however, in ~\cite{sun2020competitive} the EV could be partially charged based on how much the charging station is crowded and its power capacity constraint. 

The second stream makes two decisions for each EV: admission control (whether to accept the EV for charging) and scheduling (how to charge the EV if it is admitted). This stream of work provides guaranteed service for all admitted EVs~\cite{alinia2019online, stein2012model,sun2018eliciting, tang2014online}. Our work lies in this stream of works. In ~\cite{alinia2019online}, they propose an algorithm to maximize the social welfare considering the charging capacity with on-arrival commitment which guarantees a charging amount upon arrival of each EV; however, they only consider sum of all EVs utilities and disregard the electricity cost and only guarantee to charge the EV larger than a fraction of its demand not to charge up to its full demand.
In \cite{stein2012model}, whenever an agent is selected,
their mechanism pre-commits to charging the vehicle by its reported departure time, but maintains flexibility about when the charging takes place and at what rate. The online mechanism in~\cite{stein2012model} does not have on-arrival commitment and commits to charge the EV at least to a fraction of its demand in arbitrarily time after its arrival. 

In terms of considering the resource cost, our work is most relevant to~\cite{sun2018eliciting} because it considers the charging cost and also guarantees to charge an EV up to its demand if it is admitted. In~\cite{sun2018eliciting}, the authors propose a cost function and design a pricing function heuristically and then analyze its  competitive ratio. The derivation of the pricing function in our work is based on the well-established online primal-dual framework and solving partial differential equations. This approach leads to have smaller competitive ratio and better empirical profit ratio compared to~\cite{sun2018eliciting}.Last,~\cite{tang2014online} also guarantees charging of admitted EVs up to their demand and its objective function is minimizing the total energy cost, but ignores the capacity constraints of EV charging stations.  

\section{Conclusions}
\label{sec:concludsion}
In this paper, we developed a competitive algorithm for joint pricing and scheduling of EV charging where EVs arrive at the station in an online manner. Our theoretical results are of independent interest since it provides the first order-optimal algorithm for the integral version of the online knapsack problem with reusable resources and time-varying resource cost. Using real data traces, we demonstrated the proposed algorithms outperform the alternatives. 

An interesting future direction is to extend the results into the multi-station setting, where the EVs submit their charging request to multiple stations simultaneously, and the service provider should determine whether or not to accept the request and, if so, when (scheduling) and where (station assignment) to charge the vehicle. 

\appendix

\clearpage
\bibliographystyle{plain}
\bibliography{bibliography.bib}


\end{document}